%% file: main.tex
\documentclass[10pt]{amsart}
\usepackage{graphicx,amssymb,amsfonts,amsmath,amsthm,newlfont}
\usepackage{epsfig,url}
\usepackage{color}
\usepackage{bm}

\usepackage[sort&compress,numbers]{natbib}

\usepackage[all,2cell]{xy} \UseAllTwocells \SilentMatrices

\newtheorem{thm}{Theorem}[section]
\newtheorem{cor}[thm]{Corollary}
\newtheorem{lem}[thm]{Lemma}
\newtheorem{prop}[thm]{Proposition}
\theoremstyle{definition}

\theoremstyle{remark}
\newtheorem{rem}[thm]{Remark}

\numberwithin{equation}{section}
\newcommand{\Z}{\mathbb Z}
\newcommand{\C}{\mathbb C}

\newcommand{\R}{\mathbb R}

\newcommand{\Pro}{\mathbb P}

\newcommand{\gr}{\mathrm{gr}}

\newcommand{\Q}{\mathbb Q}

\newcommand{\eps}{\varepsilon}
\newcommand{\To}{\longrightarrow}

\newcommand{\A}{\mathbb{A}}
\newcommand{\G}{\mathbb{G}}
\newcommand{\E}{\mathcal{E}}

\newcommand{\SL}{\mathrm{SL}}

\newcommand{\mm}{\mathfrak{m} }

\newcommand{\id}{\mathrm{id} }

\newcommand{\per}{\mathrm{per}}

\newcommand{\Pe}{\mathcal{P}}

\newcommand{\comp}{\mathrm{comp}}

\def\beq{\begin{equation}}
\def\eeq{\end{equation}}
\def\bsp#1\esp{\begin{split}#1\end{split}}

\newcommand{\cN}{\mathcal{N}}

\newcommand{\cC}{\mathcal{C}}

\begin{document}
\begin{flushright}CERN-TH-2020-097\end{flushright}

\author{Francis Brown}
\author{Claude Duhr}

\begin{title}[A double integral of dlog forms which is not polylogarithmic]
{A double integral of dlog forms which is not polylogarithmic} \end{title}
\maketitle

\begin{abstract}
Feynman integrals are central to all calculations in perturbative Quantum Field Theory. 
They often give rise to iterated integrals of $d\log$-forms with algebraic arguments, which in many cases can be evaluated in terms of multiple polylogarithms. This has led to certain folklore beliefs in the community stating that  all such integrals evaluate to polylogarithms. Here we discuss a concrete example of a double iterated integral of two $d\log$-forms that evaluates to a period of a cusp form. The motivic versions of these integrals are shown to be algebraically independent from all multiple polylogarithms evaluated at algebraic arguments. 
From a mathematical perspective, we study a mixed elliptic Hodge structure arising from a simple geometric configuration in $\mathbb{P}^2$, consisting of a modular plane elliptic curve and a set of lines which meet it at torsion points, which may provide an interesting worked example from the point of view of periods, extensions of motives, and $L$-functions. 
\end{abstract}

%\input{conventions.tex}

\input{introduction.tex}

\vskip 1cm
%\begin{center}
%{\large PART I: A MIXED ELLIPTIC  MOTIVE}
%\end{center}
\input{dlogmotive.tex}
\input{modular.tex}

%\vskip 1cm
%\begin{center}
%{\large PART II: ONE-LOOP INTEGRALS AND POLYLOGARITHMS}
%\end{center}
%\input{oneloop.tex}

%\section*{acknowledgement}
%ERC `MathAm''. ``Mathemamplitudes'' in Padova.

\appendix

\input{modular_appendix.tex}

\bibliography{bib}
\bibliographystyle{JHEP}

 \end{document}

%% file: introduction.tex
% !TEX root = main.tex

\section{Physics context and summary of main results}
\label{sec:introduction}

\subsection{Feynman graphs and integrals}

Quantum Field Theory is among the main frameworks of the physics of our time, and the backbone of all computational techniques to compare theory and experiments in high-energy physics. The interactions among the quantum fields and states are encoded by correlation functions and on-shell scattering amplitudes. While it is not  in general  possible to compute these quantities exactly,  they can be expanded into a perturbative series in cases where the theory contains a small parameter. 

The perturbative series can be neatly organised in terms of Feynman graphs, and the $L$-th order in the perturbative expansion receives contributions from Feynman graphs with $L$ loops. The precise definition of a Feynman graph is not important for the purposes of this paper (see, for example,~\cite{Brown:2015fyf}). Here it suffices to say that to each Feynman graph one can associate a \emph{Feynman integral}, which depends on the dimension $d$ of space-time and is a function of the external kinematic data (e.g., the masses and momenta of all external particles).

Whenever it converges, a Feynman integral
defines a family of periods depending on kinematic parameters that is a generalisation of the notion of period in the sense of Kontsevich and Zagier~\cite{PeriodsZagier}. For algebraic values of the masses and momenta it is exactly a period in their  sense.  One can show that Feynman integrals can be promoted (at least when the masses and momenta are generic and $d$ is even) to `motivic periods' of the cohomology of a family of algebraic varieties  \cite{Brown:2015fyf}. The same is almost certainly true in all cases. 

\begin{rem}
Feynman integrals are often divergent and need to be  regularised. While various different regularisations exist, the most commonly used regularisation in physics is \emph{Dimensional Regularisation}~\cite{tHooft:1972tcz,Bollini:1972ui,Cicuta:1972jf}. Loosely speaking, it consists in replacing the space-time dimension $d$ by $D=d-2\eps$, where $\eps$ is a variable taking values in $\mathbb{C}$. One obtains in this way a meromorphic function of $\eps$~\cite{speer},
which admits a Laurent expansion around $\eps=0$. The objects of interest are the Laurent coefficients, which are then also families of  periods in the spirit of Kontsevich and Zagier~\cite{Bogner:2007mn}. In applications one is only interested in  the first few terms in the Laurent expansion, because only a finite number of Laurent coefficients contribute to the physical observable of interest (which must be finite and independent of the chosen regularisation). 
\end{rem}

%%%%%%%%%%%%%%%%%%%%%%%%%%%%%%%%%%%%%%%%%%
\subsection{Iterated integrals and multiple polylogarithms} 
It is known that large classes of Feynman integrals can be evaluated in terms of iterated integrals. Let  $X$ be a smooth $m$-dimensional complex manifold. Let $\gamma:[0,1]\to X$ be a piecewise smooth path on $X$, and let $\omega_1,\ldots,\omega_n$ be smooth  one-forms on $X$. We denote the pullback of $\omega_i$ to the interval $[0,1]$ by $dt\,f_i(t) = \gamma^\ast\omega_i$. The iterated integral of the forms $\omega_1,\ldots,\omega_n$ along $\gamma$ is defined as
\beq\bsp
\int_{\gamma}&\,\omega_1\cdots\omega_n :=\int_{0\le t_1\le t_2\le \cdots\le t_n\le 1}\gamma^{\ast}\omega_1\wedge\ldots\wedge \gamma^{\ast}\omega_n\\
&= \int_{0\le t_1\le t_2\le \cdots\le t_n\le 1}dt_1\,f_1(t_1)\cdots dt_n\,f_n(t_n)\\
&=\int_0^1dt_n\,f_n(t_n)\int_0^{t_n}dt_{n-1}\cdots\int_0^{t_2}dt_1\,f_1(t_1)\,.
\esp\eeq
More generally, an iterated integral is any linear combination of such integrals.  The
empty iterated integral (when $n = 0$) is defined to be the constant function 1.
We will only be interested in homotopy-invariant iterated integrals, i.e., linear combinations that only depend on the homotopy class of the path $\gamma$ in $X$. All iterated integrals that appear in the computation of Feynman integrals are of this type, where $X$ denotes the complex points of a smooth algebraic variety over $\overline{\Q}$,  and all iterated integrals are $\overline{\Q}$-linear combinations of integrals of  forms $\omega_i$ which are globally defined logarithmic forms on $X$ which are defined over $\overline{\Q}$.
To spell this out in more detail, Feynman integrals  typically give rise to homotopy-invariant iterated integrals of forms $\omega_i $ which are holomorphic or of the form $d\log f$ where $f$ is a rational function on $X$. Frequently,  $X$ is an open subset of (a finite covering of)  an affine  space with coordinates $x_1,\ldots, x_n$, and one  can write   
$\omega_i = d\log R_i(x_1,\ldots,x_m)$, where $R_i(x_1,\ldots,x_m)$ is a rational (or algebraic) function. 

A particularly important representative of iterated integrals of $d\log$-forms are \emph{multiple polylogarithms} (also known as hyperlogarithms), which were first introduced in the works of Poincar\'e, Kummer and Lappo-Danilevsky~\cite{Kummer,Lappo:1927} and have recently reappeared in both mathematics~\cite{GoncharovMixedTate,Goncharov:1998kja,Brown:2011ik} and physics~\cite{Remiddi:1999ew,Gehrmann:2000zt,Ablinger:2011te}. Multiple polylogarithms can be defined as
\beq\label{eq:G_def}
I(a_1,\ldots,a_n;z) = \int_0^z\frac{dt}{t-a_1}\,I(a_2,\ldots,a_n;t)\,,
\eeq
where $a_i$, $z\in\mathbb{C}$, and the recursion starts at $I(;z)=1$. If $a_n=0$ the integral in~\eqref{eq:G_def} diverges, and we define instead  
\beq
I(\underbrace{0,\ldots,0}_{n\textrm{ times}};z) = \frac{1}{n!}\log^nz\,.
\eeq
Multiple polylogarithms are well-studied in mathematics and in physics. In particular, it is well understood how to perform algebraic manipulations of multiple polylogarithms.
In addition, there are several fast numerical implementations of these functions that can be used for their evaluation at high precision~\cite{Gehrmann:2001pz,Gehrmann:2001jv,Vollinga:2004sn,Buehler:2011ev,Frellesvig:2016ske,Ablinger:2018sat,Naterop:2019xaf}. Given the good algebraic and numerical control one has over multiple polylogarithms, it is often desirable to express Feynman integrals, scattering amplitudes and correlation functions in terms of multiple polylogarithms whenever possible. 

If an iterated integral is homotopy-invariant and the functions $R_i(x_1,\ldots,x_m)$ are rational functions of the variables $x_i$ (with coefficients in $\overline{\mathbb{Q}}$ say), and if the base point of the integration path is algebraic, then one can always write the integral in terms of multiple polylogarithms evaluated at algebraic arguments. Indeed, we can use homotopy-invariance and replace the path of integration by a homotopic path along the edges of a hypercube where all but one of the variables are constant. 
However, if the $R_i(x_1,\ldots,x_m)$ are not rational, no such algorithm exists. Nonetheless, many examples of iterated integrals of $d\log$-forms with non-rational arguments that have appeared in physics can be evaluated in terms of multiple polylogarithms (see, e.g.,~\cite{Heller:2019gkq}). This has led to folklore conjectures in the physics community that \emph{every} (iterated) integral of $d\log$-forms with algebraic arguments can be expressed in terms of multiple polylogarithms evaluated at algebraic points (at least in principle). The purpose of this paper is to show that this is false by  providing an explicit example of a double iterated integral of $d\log$-forms which cannot be expressed in terms of any linear combination of multiple polylogarithms evaluated at algebraic arguments (assuming the standard period conjecture). We highlight the implications for quantum field theory below, after a brief technical summary of our results.  

\subsection{Summary of results}
Let $\rho = -e^{2\pi i/3} = e^{-i\pi/3}$, and $\bar{\rho}$ denote its complex conjugate.
In \S\ref{sec:two_periods} we consider the iterated integrals:
\begin{align}
\label{eq:IE_def}
I_{\E} & = 2 \, \mathrm{Re} \int_{-1\leq x_1 \leq x_2\leq \infty}  \frac{1}{\rho- \overline{\rho}}   \,  d\log \left(\frac{x_1-\rho }{x_1- \overline{\rho}
  } \right)  \wedge  d\log\left(\frac{\sqrt{1+x_2^3}+1}{\sqrt{1+x_2^3}-1}\right)\,,\\
\label{eq:I_def}
I&=\int_{2\le x_1\le x_2\le\infty}\frac{1}{\rho-\bar{\rho}}\,d\log\left(\frac{x_1-\rho}{x_1-\bar{\rho}}\right)\wedge d\log\left(\frac{\sqrt{1+x_2^3}+1}{\sqrt{1+x_2^3}-1}\right)\,.
\end{align}

In order to interpret these integrals geometrically, consider the algebraic curve in $\mathbb{P}^2$ defined by the equation
\beq\label{eq:E_equation}
zy^2 = z^3+x^3\ .
\eeq
It defines an elliptic curve $\E$, so in particular  it is not possible to find any change of variables such that the argument of the logarithm in the integrands in~\eqref{eq:IE_def} and~\eqref{eq:I_def} becomes rational.  By constructing the underlying `motives' of these integrals (we shall use the word `motive' loosely to mean an object in a category of realisations which arises from the cohomology of an algebraic variety) and proving that they contain a non-trivial mixed elliptic extension, we prove in Corollary~\ref{cor:main_1} and~\ref{cor:main_2}  that the motivic versions $I_{\E}^{\mathfrak{m}}$ and $I^{\mathfrak{m}}$ of these integrals are algebraically independent from all motivic polylogarithms at algebraic points. It then follows from a version of  Grothendieck's period conjecture that $I_{\E}=\per(I_{\E}^{\mathfrak{m}})$ and $I=\per(I^{\mathfrak{m}})$ cannot be expressed in terms of multiple polylogarithms evaluated at any algebraic argument. 

The obstruction to being polylogarithmic is the same for both $I_{\E}^{\mathfrak{m}}$ and $I^{\mathfrak{m}}$. More precisely, we show in \S\ref{sec:I_pol} that there is a linear combination $I^{\mm}_{\mathrm{Pol}}$ of $I_{\E}^{\mathfrak{m}}$ and $I^{\mathfrak{m}}$: 
\beq\label{eq:IPol}
I^{\mm} = \frac{1}{6}\, I^{\mm}_{\mathcal{E}} +   I^{\mm}_{\mathrm{Pol}}\,,
\eeq
where $I^{\mm}_{\mathrm{Pol}}$ is a motivic period of a mixed Artin-Tate object, which numerically evaluates to a linear combination of  dilogarithms and logarithms. The obstruction itself is an extension of $H^1(\mathcal{E})$ by a certain Dirichlet motive $\Q_{\chi}(-2)$. The non-triviality of this extension is precisely detected  by the non-vanishing of the integral $I_{\E}$.   Furthermore, Beilinson's conjecture~\cite{Bec,Be} then predicts that $I_{\E}$, which is essentially the regulator, is proportional to the (non-critical) value at $2$ of the $L$ function of $\mathcal{E}$. Indeed, this is what we find numerically, and  could almost certainly be proven rigorously using the theory of iterated integrals of Eisenstein series (i.e., multiple modular values) as we now explain.

 A key point is that $\E$ defined by~\eqref{eq:E_equation} admits a modular parametrisation. Let $\Gamma(N)\subset \SL_2(\mathbb{Z})$ be the principal congruence subgroup of level $N$, $\mathfrak{H} = \{\tau\in\mathbb{C} : \textrm{Im }\tau>0\}$ the complex upper half-plane, and $Y(N)$  the modular curve obtained by taking the (orbifold) quotient of $\mathfrak{H}$ by the usual action of $\Gamma(N)$ via M\"obius transformations. There is an isomorphism $\varphi: Y(6) \to \E\backslash\mathcal{C}$, where $\mathcal{C}$ denotes a finite set of points.  The pullback $\varphi^* \omega$ of a logarithmic differential form $\omega$  on $\E$ with poles along $\mathcal{C}$ can be identified with a modular form of weight two for $\Gamma(6)$. The  holomorphic differential on $\E$ pulls back to the unique (normalised) cusp form of weight two, whereas $d\log$-forms pull back to linear combinations of Eisenstein series and this cusp form. As a result the integral $I$ in~\eqref{eq:I_def} can be expressed  (\S\ref{sec:modular}) as  a   double iterated integral
\beq\label{eq:I_Eisenstein}
I=\frac{1}{\rho-\bar{\rho}}\int_{0\le t_1\le t_2\le\infty}\frac{dt_1\wedge dt_2}{(2\pi)^2}\,E_1(it_1)\,E_2(it_2)\,,
\eeq
where $E_1(\tau)$ and $E_2(\tau)$ are certain  Eisenstein series of weight two for $\Gamma(6)$. In~\cite{Brown:mmv} it was shown that double iterated integrals of Eisenstein series of small weight for the full modular group $\Gamma(1)$ evaluate to multiple zeta values,  and periods of simple extensions of motives of cusp forms for $\Gamma(1)$. The latter include  non-critical  $L$-values  of cusp forms (amongst other quantities) and  first   appear when the sum of the modular weights of the two Eisenstein series is twelve, because the first cusp form for $\Gamma(1)$ has weight twelve. Since $\Gamma(6)$ has genus one, the first cusp form already appears in weight two:
\beq
f(\tau) = \eta(\tau)^4\,,
\eeq
where $\eta(\tau)$ is the Dedekind $\eta$ function. 
 By the general theory, we therefore expect the double iterated integral in~\eqref{eq:I_Eisenstein} to evaluate to a linear combination of multiple polylogarithms evaluated at sixth roots of unity and the value at $2$ of the completed $L$ function of the cusp form $f$: 
 \beq
\Lambda(f,2) = \int_{i\infty}^0d\tau\,f(\tau)\,\tau = 0.85718907492991773071685111\ldots\, 
\eeq
Using the PSLQ algorithm, we find (with $I_{\mathrm{Pol}} = \per(I^{\mm}_{\mathrm{Pol}})$):
\begin{align}
I_{\E} &= -4 \, \pi \sqrt{3}\,  \Lambda(f,2)\,,\\
I_{\mathrm{Pol}} &= \frac{5}{\sqrt{3}}\,\textrm{Cl}_2\left(\frac{\pi}{3}\right)\,,
\end{align}
where $\textrm{Cl}_2\left(\frac{\pi}{3}\right) = \textrm{Im }\textrm{Li}_2(e^{i\pi/3})$. These evaluations could be proven rigorously with a more detailed analysis: the first should follow from an application of a version of the Rankin-Selberg method to double Eisenstein integrals, as was done in \cite{Brown:mmv2} in level $1$. The second could be proven using, for example,  the theory of multiple elliptic polylogarithms and unipotent completion of the motivic fundamental group of the universal elliptic curve. Since these computations are quite lengthy and technical and are not required for the main  point of this paper, they are not presented here.

\subsection{Implications for Quantum Field Theory}
The integrals $I_{\E}$ and $I$ are explicit examples of integrals of $d\log$-forms that cannot be evaluated in terms of multiple polylogarithms at algebraic points. Here we discuss some implications  for perturbative Quantum Field Theory, because these integrals are concrete counter-examples to certain folklore  beliefs in the community stating that  all integrals of $d\log$-forms evaluate to polylogarithms. 
%In short,  In the remainder of this section we present some possible implications for physics.

{\bf The role of the integration cycle.} An integrand that can be written in $d\log$-form with algebraic arguments is  insufficient for an integral to evaluate to multiple polylogarithms. Whether an integral evaluates to polylogarithms is not determined by the integrand alone, but also the integration cycle, which plays an important role. To illustrate this point, let us return to~\eqref{eq:IPol}:  whilst $I^{\mm}_{\mathcal{E}}$ and $I^{\mm}$ are examples of motivic periods  which  are algebraically independent from motivic polylogarithms, the period $\mathrm{per}\,(I^{\mm}_{\mathrm{Pol}})$ does evaluate to multiple polylogarithms, even though the periods of all three objects are integrals involving the same 
$d\log$-form in the integrand.\footnote{The fact that we only consider the real part of $I_{\mathcal{E}}$ immaterial here. See \S\ref{sec:Sec4} for the precise definition of the integration cycles used to compute $I_{\mathcal{E}}$ and $I$.} Thus, looking at the integrand alone is insufficient to decide if an integral can be evaluated in terms of multiple polylogarithms, and the choice of the integration cycle is  important. 

{\bf Canonical differential equations in \emph{d}log-form.} It follows from our examples that it is not clear that families of (dimensionally-regularised) Feynman integrals that satisfy a system of linear differential equations in `canonical $d\log$-form' (cf.~\cite{Henn:2013pwa}) can always be evaluated in terms of polylogarithms.  Instead, other classes of iterated integrals -- such as iterated integrals of modular forms or multiple elliptic polylogarithms -- may also show up even in the case of a differential equation in `canonical $d\log$-form'. This can happen whenever the $d\log$-forms involve algebraic arguments that cannot be rationalised via a suitable parametrisation of the external kinematic data (see, e.g.,~\cite{Besier:2018jen,Besier:2019kco} for a review). This situation is known to occur for example in Feynman integrals contributing to two-loop QED corrections to Bhabbha scattering~\cite{Henn:2013woa,Festi:2018qip} as well as for the two-loop mixed QCD-QED corrections to the Drell-Yan process~\cite{Bonciani:2016ypc,Besier:2019hqd}. In~\cite{Heller:2019gkq} it was shown that in four space-time dimensions these results can be expressed in terms of multiple polylogarithms depending on complicated algebraic functions of the external kinematic data. Our examples from the previous section show that this is not the rule, and there is no reason why the same should be true for other integrals that involve $d\log$-forms depending on square roots that cannot be rationalised, nor is there a reason why it should be true for the higher orders in the Laurent expansion in the dimensional regulator $\eps$ for the integrals considered in~\cite{Henn:2013woa,Bonciani:2016ypc,Heller:2019gkq}.

{\bf Planar $\mathcal{N}$ = 4 super Yang--Mills theory.} The examples of the previous section may also have implications on conjectures about the analytic structure of certain special Quantum Field Theories, like for example the planar $\cN=4$ super Yang-Mills theory. In~\cite{Arkani-Hamed:2016byb} it was argued that scattering amplitudes in this theory for certain assignments of the quantum numbers of the external states -- the so-called \emph{maximally-helicity-violating (MHV)} and \emph{next-to-MHV (NMHV)} amplitudes -- can be expressed in terms of polylogarithms for any number of loops and external particles. A central point in the argument is a (conjectural) procedure to write the loop integrand of these amplitudes in a form which only involves $d\log$-forms. For up to seven external particles, the arguments of these $d\log$-forms are rational functions obtained from cluster algebras (of finite type) associated to certain Grassmannian spaces, see, e.g.,~\cite{Golden:2013xva,Golden:2014xqa,Drummond:2017ssj,Drummond:2018dfd,Drummond:2019cxm,Golden:2019kks,Mago:2019waa,Arkani-Hamed:2019rds,Henke:2019hve,Arkani-Hamed:2020cig}. 
It is also known that two-loop MHV amplitudes for any number of external particles can be expressed in terms of polylogarithms~\cite{CaronHuot:2010ek,Golden:2014xqa}.
Starting from eight external particles, however, two-loop NMHV and three-loop MHV amplitudes involve $d\log$-forms with algebraic arguments~\cite{Prlina:2017azl,Prlina:2017tvx}. It is currently not known if one can parametrise the external kinematic data in a way which would rationalise these algebraic arguments (though it is known that in some cases all contributions from these algebraic arguments cancel in the full amplitude~\cite{Bourjaily:2019igt}). As a consequence, conjectures stating  that  MHV and NMHV amplitudes in planar $\cN=4$ Super Yang-Mills are always expressible in terms of polylogarithms should be taken with a pinch of salt: there is no firm supporting evidence, nor a counterexample, for this conjecture beyond seven particles.

%%%%%%%%%%%%%%%%%%%%%%%%%%%%%%%%%%%%%%

\subsubsection*{Acknowledgements} The authors  thank Matija Tapu\v{s}kovi\'{c} for discussions and Johannes Broedel, Andrew McLeod and Lorenzo Tancredi for comments on the manuscript. The authors would also thank the organisers of the `Mathemamplitudes' workshop in Padova, where the ideas presented in this paper were first discussed, for the invitation to the workshop. 
This project has received funding from the European Research Council (ERC) under the European Union's Horizon 2020 research and innovation programme (grant agreement nos. 724638 and 637019).

%% file: dlogmotive.tex
% !TEX root = main.tex

\section{Geometry \& Setup}

\subsection{Geometry}
We first consider the geometric situation underlying the integrals   \eqref{eq:IE_def} and \eqref{eq:I_def}. 
Let $k = \Q(\sqrt{-3}) \subset \C$ and let $\rho = - e^{2\pi i/3} \in k$.   Let $\mathcal{E}\subset \Pro^2$ denote the elliptic curve defined in projective coordinates 
$(x:y:z)$ by $zy^2=x^3+z^3$. Let $A\subset \Pro^2$ denote the divisor over $\Q$ given by the  union of the loci 
$$x^2 - xz+z^2 =0 \quad  , \quad  y-z=0 \quad, \quad y+z=0 \ .$$
Its extension of scalars $A_{k}= A \times_{\Q} k$  is a union of four lines $y= \pm z$, $ x= \rho z$, $x = \overline{\rho}z$ which cross normally.   Let $B$ denote the union of $\mathcal{E}$ with the line $x=2z$.

The elliptic curve $\mathcal{E}$ meets $x=\rho z$   at the point $(\rho:0 :1)$ with multiplicity two and at the point at infinity $\infty= (0:1:0)$, and similarly with $\rho$ replaced with $\overline{\rho}$.  It meets  the line $y=z$   at the point $(0:1:1)$  with multiplicity three, and the line  $y=-z$  at $(0:-1:1)$ with multiplicity three as well. 
 \begin{figure}[h!]
{\includegraphics[height=6cm]{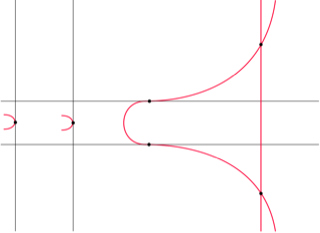}} 
\put(-236,-10){$x=\rho $}
\put(-192,-10){$x= \overline{\rho} $}
\put(-55,-10){$x=2$}
\put(2,95){$y=1$}
\put(-30,160){$\mathcal{E}$}
\put(-40,132){$(2,3)$}
\put(-40,32){$(2,-3)$}
\put(2,63){$y=-1$}
\put(-178,72){$\tiny{(\overline{\rho},0)}$}
\put(-222,72){$\tiny{(\rho,0)}$}
\put(-135,105){$(0,1)$}
\put(-135,50){$(0,-1)$}
\caption{A local picture of the  divisors $B$ (red) and the four components of  $A_k$ (black) in affine coordinates $(x,y)= (x:y:1)$. }
\end{figure}
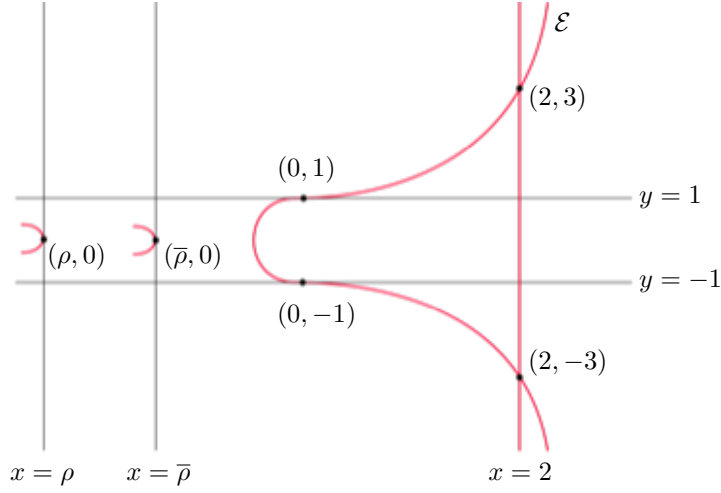
\noindent
It follows that the points $( \rho,0 )$, $(\overline{\rho},0)$ (and also $(-1,0)$), are  torsion points of  order 2 in the group $\mathcal{E}(k)$, and  the points $(0, 1), (0, -1)$ have order 3. The points $(2,3)$, $(2,-3)$ have  order 6 since, for example,    $x+z=y$ intersects $\mathcal{E}$ at $(-1,0)$, $(0,1)$ and $(2,3)$. 

From now on let us denote by $\Sigma$  the following set of five torsion points
  \begin{equation} \label{SigmaDef} \Sigma=\{ \ (\rho:0:1) \ , \  (\overline{\rho}:0:1) \ , \  (0: 1: 1) \ , \ (0:- 1: 1) \ , \  \infty  \ \} 
  \end{equation} 
and write $\mathcal{E'} = \mathcal{E}\backslash (\mathcal{E} \cap A)$. Thus $\mathcal{E}'_k= \mathcal{E}_k \backslash \Sigma$, where $\mathcal{E}_k = \mathcal{E} \times_{\Q} k$ and likewise for $\mathcal{E}'.$
Let us also denote by $\pi_x,\pi_y$ the following two   projections 
\begin{equation} \label{pixy} \pi_x \ : \   \Pro^2 \backslash \{ (0:1:0)\}  \rightarrow \Pro^1 \quad \ , \quad \pi_y  \ : \ \Pro^2 \backslash \{ (1:0:0)\}   \rightarrow \Pro^1 \end{equation}
which map $(x:y:z)$ onto $(x:z)$ or $(y:z)$ respectively. Each   restricts  to a  projection from  $ \mathcal{E} \backslash\infty$ to the affine lines $\A^1$ defined by $z=1$. 

\begin{rem} The most arithmetically interesting motives typically arise from  singular configurations, and the above situation is a case in point.  In our example, it happens that the components of the divisor  $\mathcal{A} $  are fibers of $\pi_x$ or $\pi_y$ over  points where their restrictions to $\mathcal{E'}$ fail  to be \'etale.
For an elliptic curve  in Weierstrass form $y^2=  x^3 +ax+b$, this locus is given by
$$x^3+ax+b=0  \quad \hbox{ and } \quad  \  27 y^4 -54 \, b\,  y^2 + (  4 a^3 +27b^2)=0\ ,$$
which in our situation leads to $y=\pm 1$ and $x= \rho, \overline{\rho}, -1.$ This remark may  provide a  way to generate other interesting examples. 
\end{rem}

\subsection{Modular parametrisation} 
The elliptic curve $\E$ admits the following explicit  modular parametrisation by $\Gamma(6) \backslash \mathfrak{H}$.  Let $\tau$ be in the upper-half plane $\mathfrak{H}:= \{\tau\in \C: \mathrm{Im}(\tau)>0\}$ and let
$\eta(\tau)$ denote the Dedekind $\eta$ function, 
\beq
\eta(\tau) = e^{i\pi\tau/12}\prod_{n=1}^\infty(1-e^{2\pi i\tau n})\,
\eeq
which satisfies $\eta(1+ \tau) = e^{\frac{i \pi}{12}} \eta(\tau) $ and $\eta(-\tau^{-1}) = \eta(\tau) \sqrt{\frac{\tau}{i}}. $

Consider the following $\eta$-quotients:
\beq\bsp\label{eq:x6y6_def}
x_6(\tau)&\, = \frac{\eta(2\tau)\,\eta(3\tau)^3}{\eta(\tau)\,\eta(6\tau)^3}\,,\\
y_6(\tau) &\,= \frac{\eta(2\tau)^4\,\eta(3\tau)^2}{\eta(\tau)^2\,\eta(6\tau)^4}\,, 
\esp\eeq
which are modular invariant for $\Gamma(6)$. They satisfy the following relation, as can be checked  by computing the first few Fourier coefficients: 
\beq
y_6(\tau)^2 = 1+x_6(\tau)^3\,.
\eeq
Let $\varphi:\mathfrak{H} \to \mathbb{P}^2(\mathbb{C})$ denote the map  $\tau \mapsto (x_6(\tau):y_6(\tau):1)$. Its image is clearly contained in $\mathcal{E}(\C).$ The group $\Gamma(6)$ has twelve cusps, 
given by the classes of the points $(a:c)\in \Pro^1(\Q)$,  where $a,c$ are integers modulo the relation $(a:c) \sim (a':c')$ if $a,a'$ and $c,c'$ are congruent modulo $6$. Here is a set of representatives:
$$ 0  \ , \  1/3 \ , \   1/2  \ ,  \  2/3 \ , \   1   \ , \   3/2 \ , \    2 \ , \   5/2 \ , \   3 \ , \  4  \ , \  5 \ , \   i \infty  \ .$$
Using the automorphy properties of $\eta(\tau)$ under $\mathrm{SL}_2(\Z)$ and the fact that every  cusp is equivalent to $ i\infty$ under the action by M\"obius transformations of $\mathrm{SL}_2(\Z)$, one  checks that $\varphi$ extends continuously to the cusps where it takes, respectively, the following values:
$$  (2:3:1)  \  , \  (0:1:1)  \  ,  \  (\rho:0:1)   \ ,  \   (0:-1:1) \ , \  (-2 \overline{\rho}: -3:1)   \  , \ (-1:0:1)$$
$$    (-2\rho:3:1)  \ , \   (\overline{\rho}:0:1)  \ ,  \  (2 : -3:1) \  , \  (-2\overline{\rho}:3:1) \ ,\    (-2\rho:-3:1) \ , \  \infty= (0:1:0)\,.$$
%\claudecomment{I updated the values for the cusps. They correspond to what is in my notes and in my Mathematica notebook.}
Let $\cC\subset \mathcal{E}(k)$ denote these 12 points. It follows  that $\varphi$ has degree one on  $\Gamma(6) \backslash \mathfrak{H}$  and hence induces an isomorphism $\varphi: \Gamma(6) \backslash \mathfrak{H}\overset{\sim}{\rightarrow} \mathcal{E}(\C)\backslash \cC$ (e.g., ~\cite{Yang:2006kt}).

Observe that all the sets of special points on $\mathcal{E}$ in the discussion of the previous paragraph, and in particular the set $
\Sigma$, are contained in the set of cusps $\cC.$

The pull-back of the holomorphic one-form $-3dx/y$ under $\varphi$ is 
\beq\label{eq:holomorphic_pullback}
\varphi^*\left(-3\frac{dx}{y}\right) = 2\pi i\, d\tau\,f(\tau)\,.
\eeq
where, writing $q= e^{2\pi i \tau/6}$,
\beq
f(\tau) = \eta(\tau)^4 = q- 4 q^7 + 2q^{13} + 8 q^{19}- 5 q^{25}+ \ldots \,,
\eeq
  is the unique normalised  cusp form  of weight two on $\Gamma(6)$.
The pull backs of  logarithmic differentials of the third kind on $\mathcal{E}_k =\mathcal{E}\times_{\Q}k$ with poles along $\cC$ can be expressed as a $k$-linear combination of $f(\tau)$ and Eisenstein series  of weight two. 

\subsection{The $L$-function}
Let $L(f,s)=\sum_{n\geq 1 }  \frac{a_n}{n^s} =1 - \frac{4}{7^s} +\frac{2}{13^s}+\ldots$ denote the $L$-function associated to $f$, where $f(q) = \sum_{n\geq 1} a_n q^n$.
Its  completed version is 
$$\Lambda(f,s) =3^s     \pi^{-s} \Gamma(s) L(f,s)\ ,$$
and satisfies the functional equation $\Lambda(f,s)= \Lambda(f,2-s)$. One easily computes its numerical value at the  non-critical point $s=2$:
$$\Lambda(f,2)=   0.85718907492991773071685111\ldots \ .$$

\subsection{Interpretations}
We shall interpret the integrals in~\eqref{eq:IE_def} and~\eqref{eq:I_def} in several different ways, as: 
\begin{enumerate}
\item Multiple modular values, i.e., iterated integrals of modular forms of weight 2  along geodesic paths between cusps on the modular curve $\Gamma(6) \backslash \mathfrak{H}$, see  \S\ref{sec:modular}. These are periods of the relative completion (in this case, the unipotent completion in fact suffices) of the fundamental groupoid of the modular curve  between tangential base points.
\item  Multiple elliptic polylogarithms, i.e., iterated integrals   on the elliptic curve $\mathcal{E}'(\C)$ (periods of the unipotent completion of its fundamental groupoid).
\item Periods of the mixed Hodge structure associated to a  specific configuration of algebraic varieties such as the one described above. 
\end{enumerate}
 The relation between $(1)$ and $(2)$ comes about because the elliptic curve $\mathcal{E}$  is modular, and because unipotent completion is a special case of relative completion.  In order to understand the relation between $(2)$ and $(3)$, recall Beilinson's general construction which  associates a motive to the unipotent fundamental group.

 Recall that $\mathcal{E}'= \mathcal{E} \backslash \Sigma.$
  Iterated integrals of length two  between two distinct points $p,q$   of $\mathcal{E}'$ are periods of  $H^2 (\mathcal{E'} \times \mathcal{E'} , Y)$ where $Y=( \{p\} \times \mathcal{E'}) \cup \Delta \cup (\mathcal{E'} \times \{q\})$ and $\Delta$ is  the diagonal.  
     The projections \eqref{pixy} together  define a morphism 
$$ \pi_x \times \pi_y\ :  \  \mathcal{E'} \times \mathcal{E'}  \ \To \   \A^1 \times \A^1 \subset \Pro^2$$
which maps the diagonal $\Delta$ to the  embedded curve $\mathcal{E'}\subset \Pro^2$. Under this morphism, the  variety  $ \mathcal{E'} \times \mathcal{E'}$ maps to the complement of some lines in $\Pro^2$ which contains $A_k$, and $Y$  maps to a divisor which contains the   elliptic curve  $\mathcal{E}'\subset \Pro^2$ together with some further lines which are parallel to the coordinate axes. 
In this way, we are  naturally led to consider the relative cohomology of  geometric configurations in $\Pro^2$  very similar  to the one  described in the previous paragraph.
%The morphism $\pi_x \times \pi_y$   presumably gives rise to equivalences of motivic periods such as those studied here. 
%The associated  motivic periods can in this way be shown to be equivalent to certain motivic periods  of the unipotent completion of the fundamental groupoid  of $\mathcal{E'} $$(2)$ \claudecomment{What is $\mathcal{E'} $$(2)$?}

\subsection{The `motive'}
The divisor $A\cup B$, even after extension of scalars to $k$, is not normal crossing, since $x=\rho z$, $x= \overline{\rho}z$, $\mathcal{E}$ and $x=2z$ all meet at the point $\infty=(0:1:0)$. 
Therefore let 
$\pi: P \rightarrow \Pro^2$
denote the blow-up of $\Pro^2$ at the point $\infty$, and let $\widetilde{A}$ denote the strict transform of $A$, and $\widetilde{B}$ the total transform of $B$. 

The divisor $\widetilde{A}_k$ is simple normal crossing, and consists of the strict transforms of 
$y=z, y=-z$, which meet at $(1:0:0)$, and $x=\rho z$, $x=\overline{\rho}z$, which do not meet.  Their mutual intersections over $k$  are $(1:0:0)$ together with: 
\begin{equation} \label{4points} ( \rho: 1:1) \ , \ ( \rho: -1:1)\ , \ ( \overline{\rho}: 1:1) \ , \ (\overline{\rho}:-1:1)\ .
\end{equation}

\begin{figure}[h!]
{\includegraphics[height=5cm]{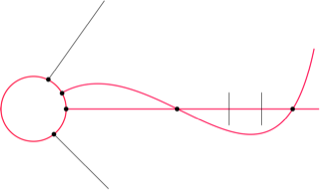}} 
\put(-2,100){$D_1$}
\put(-62,65){$D_2$}
\put(-162,10){$x=\overline{\rho}$}
\put(-170,120){$x=\rho$}
\put(-253,60){$D_3$}
\caption{A local picture of the blow-up $P$ at the point at infinity.  Not all intersections are shown. The divisor $\widetilde{B}$ is pictured in red.
The divisor $D_1$  (resp. $D_2$) is an open in the strict transform of $\mathcal{E}$, (resp. $x=2z$), and $D_3$ is an open in the exceptional locus.  }
\end{figure}
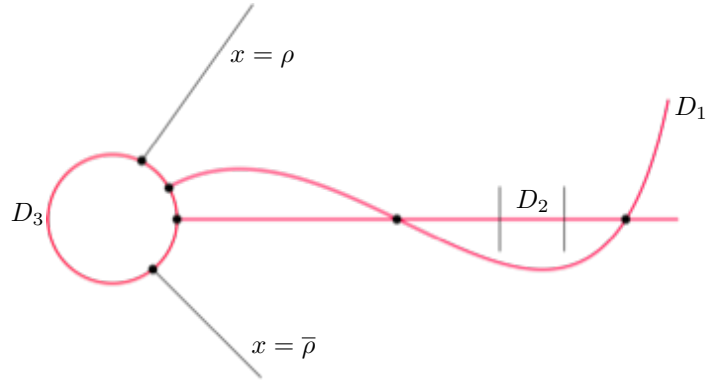

The object of study will be the relative cohomology:
\begin{equation}\label{Mdef}  M = H^2(P \backslash \widetilde{A} \  , \   \widetilde{B} \backslash \widetilde{B} \cap \widetilde{A})\ .
\end{equation} 
It defines an object in a category of systems of  realisations (we will mostly be concerned with Betti and de Rham cohomology) over $\Q$, since the pair $ (P \backslash \widetilde{A} \  , \   \widetilde{B} \backslash \widetilde{B} \cap \widetilde{A})$ is defined over  the rationals. 
 Alternatively, one can retrieve $M$ from  the object  
 \begin{equation} \label{Mkdef} 
 M_k =   H^2(P_k \backslash \widetilde{A}_k \  , \   \widetilde{B}_k \backslash \widetilde{B}_k \cap \widetilde{A}_k) \ ,
 \end{equation} 
in a  system of realisations over $k$, together with an action of $\mathrm{Gal}(k/\Q)$. 

\begin{rem} Note that $\widetilde{A} \cup \widetilde{B}$ is not simple normal crossing because the curve $\mathcal{E}$ meets $\widetilde{A}$ at points with multiplicity $>1$. One can, if one chooses, blow up $P$  successively at the points of intersection of $\mathcal{E}$ with $\widetilde{A}$ to obtain a simple normal crossing divisor, and define the motive to be the cohomology of this new space minus the total transform of $\widetilde{A}$, taken relative to the strict transform of $\widetilde{B}$.  Since the exceptional divisors are ultimately removed in this procedure, this has no effect on  the cohomology and the resulting motive is identical to $M$. 
\end{rem}

\section{Calculation of the `motive' $M$} The following calculations are valid in any reasonable cohomology theory (e.g., singular, or algebraic de Rham). 
We shall work in a version $\mathcal{H}$  of Deligne's category of realisations considered in particular in \cite{brownnotesmot} \S2. 
 The object  \eqref{Mdef}  defines an object (called simply a `motive' by abuse of terminology) in  $\mathcal{H}$ which is  a triple
  $$(M_B, M_{dR}, \comp)$$
  where  $M_B =  H^2 (  ( P \backslash \widetilde{A} )(\C) \  , \   (\widetilde{B} \backslash \widetilde{B} \cap \widetilde{A})(\C);\Q)$ is the singular cohomology group,   and $M_{dR} =  H_{dR}^2(P \backslash \widetilde{A} \  , \   \widetilde{B} \backslash \widetilde{B} \cap \widetilde{A};\Q)$ is the algebraic de Rham cohomology group associated to \eqref{Mdef}. They are $\Q$-vector spaces equipped with a weight  filtration  $W$  (and Hodge filtration $F$ on $M_{dR}$) and  $\comp: M_{dR} \otimes \C \cong M_B\otimes \C$ is the comparison isomorphism. The space $M_B$ is also equipped   with a real Frobenius  involution $F_{\infty}$ induced by complex conjugation on complex points.  This data is subject to a number of constraints - in particular $M_B$ has a   $\Q$-mixed Hodge structure.

We shall show that the  weight-graded (semi-simple) object associated  to $M$ is built out of the following simple objects  in the category $\mathcal{H}$: 
$$\Q(-2) \ ,  \ \Q_{\chi}(-1) \ , \   H^1(\mathcal{E})\ ,  \Q(-1) \ ,  \ \Q(0) \ , $$
where $\Q_{\chi}$ is a certain Dirichlet motive (defined below), $\Q(0)= H^0(\mathrm{Spec}(\Q))$, and  $\Q(-1) = H^1(\G_m)$.  
 Here, and subsequently,  $H^i(X)$ is shorthand for the object 
 $$H^i(X):= (H^i(X(\C)), H^i_{dR}(X), \mathrm{comp})$$ in the category $\mathcal{H}$ if $X$ is defined over $\Q$.  The notation $M(-n)$ stands for the Tate twist $M \otimes \Q(-1)^{\otimes n}$ as usual.

\subsection{An Artin motive}
Consider the rank two  Artin motive of weight zero:
$$\mathcal{A}=H^0\left(\mathrm{Spec}\, k\right)   \  . $$
   Via the morphism $\mathrm{Spec}\, k\rightarrow  \mathrm{Spec}\, \Q$, it splits  in $\mathcal{H}$ into a direct sum  $\mathcal{A} = \Q \oplus \Q_{\chi}$
 of the trivial motive and a Dirichlet motive $\Q_{\chi}$.  After extending coefficients to $k$, the motive $\Q_{\chi} \otimes_{\Q} k $  (in the category $\mathcal{H}\otimes_{\Q} k$)  becomes isomorphic to the trivial motive $\Q\otimes_{\Q} k$. 
We can interpret the Tate twist $\mathcal{A}(-1)$ of $\mathcal{A}$  as the object: 
\begin{equation} \label{Aminus1} \mathcal{A} (-1) =  H^1(\Pro^1 \backslash \{V(x^2-x+1) ,\infty\})  \  . \end{equation}
Note that $\Pro^1 \backslash \{V(x^2-x+1) ,\infty\})\cong \Pro^1 \backslash \{V(x^3+1)\})$
via the map $y\mapsto \frac{y+1}{2-y}$.
% which sends $(\infty, \rho, \overline{\rho})$ to $(-1, \rho, \overline{\rho})$. 

The de Rham realisation $\mathcal{A} (-1)_{dR}$ of \eqref{Aminus1}  is the 2-dimensional $\Q$-vector space generated by the cohomology classes of the forms
$$\frac{dx}{x^2-x+1} \quad    \quad  \hbox{ and }  \quad   \quad  \frac{(2x-1) \, dx}{x^2-x+1} = d\log (x^2-x+1) \ ,   $$
which generate $\Q_{\chi}(-1)_{dR}$ and $\Q(-1)_{dR}$  respectively.  The  Betti realisation $\mathcal{A} (-1)_{B}$ is the 2-dimensional $\Q$-vector space generated by classes of loops $\sigma_{\rho}, \sigma_{\overline{\rho}}$  winding positively around $\rho$, and $\overline{\rho}$ respectively.  Since the real Frobenius $F_{\infty}$ acts via 
$F_{\infty} [ \sigma_{\rho}] = - [  \sigma_{\overline{\rho}}]$, and $F_{\infty}$ acts on   $\Q(-1)_B$ by 
$-1$, we deduce that  $\Q_{\chi}(-1)_B$  is spanned by $[\sigma_{\rho}]- [\sigma_{\overline{\rho}}]$, which is $F_{\infty}$-invariant.  The period matrix of $\Q_{\chi}(-1)$ is therefore the $(1\times 1)$ matrix
$$ \left( \int_{\sigma_{\rho} - \sigma_{\overline{\rho}}} \, \frac{ dx}{x^2-x+1}  \right) = \left( \frac{4 \pi i }{3}\left(\overline{\rho}-\rho\right)  \right) =   \left( - \frac{4 \pi }{\sqrt{3}  } \right) \  ,   $$
with respect to the above bases. 

\begin{rem} For computations, it  will be  convenient to consider not the object $M$ in the category $\mathcal{H}$ but rather the object $M_k$  defined in   \eqref{Mkdef}
in a category $\mathcal{H}_k$ of realisations of objects over $k$. The object $M$ is retrieved from $M_k$ together with  the data of an action of $\mathrm{Gal}(k/\Q)$ on the components of $M_k$.  The essential difference with $\mathcal{H}$ is that the category $\mathcal{H}_k$  consists of triples $( (M_{\sigma})_{\sigma}, M_{dR}, (\comp_{\sigma})_{\sigma})$ where $M_{\sigma}$ are $\Q$-vector spaces for every embedding $\sigma:k \hookrightarrow \C$, and $M_{dR}$ is a finite-dimensional $k$-vector space, satisfying similar compatibilities to those considered before. There are two isomomorphisms
$\mathrm{comp}_{\sigma}: M_{dR} \otimes_{k,\sigma} \C \overset{\sim}{\rightarrow} M_{\sigma}\otimes_{\Q} \C$, one for each embedding $\sigma$
 of $k$. The natural functor $\mathcal{H} \rightarrow \mathcal{H}_{k}$
 sends $(M_B, M_{dR}, c)$ to  $((M_B)_{\sigma}, M_{dR} \otimes k, (c \otimes \sigma)_{\sigma})$.
 The simple objects 
$\Q(-n), H^1(\mathcal{E})$ in $\mathcal{H}$ correspond to objects in $\mathcal{H}_k$ with a trivial $\mathrm{Gal}(k/\Q)$ action, but  the  object $\Q_{\chi}$ corresponds to the trivial object $\Q$  in $\mathcal{H}_k$  equipped with a non-trivial action of
$\mathrm{Gal}(k/\Q)$.  This action is given by the usual (semi-linear) Galois  action on its de Rham component $M_{dR}=k$, and permutes the two Betti  components $M_{\sigma}$ and the two maps $c_{\sigma}$. 
\end{rem}

\subsection{Preliminary calculations} 

\begin{lem} \label{lem: prelim}
As an object of $\mathcal{H}$, we have:
$$ H^n(P \backslash \widetilde{A}) = \begin{cases}
\Q(0) \qquad \qquad\qquad \quad  \hbox{ if } n=0  \ , \\     \Q(-1)\oplus \Q_{\chi}(-1)  \qquad \hbox{ if } n=1  \\   \Q(-2)\oplus \Q_{\chi}(-2)  \qquad \hbox{ if } n=2   \ ,  
\end{cases} 
$$
and vanishes for $n\geq 3$. 
\end{lem} 

\begin{proof}   We work for now  with objects in $\mathcal{H}_k$ equipped  with a $\mathrm{Gal}(k/\Q)$-action. 
From the formula for the cohomology of a blow-up,  the  odd degree cohomology of $P_k$ vanishes, and one has (as objects of $\mathcal{H}_k$)
$$H^0(P_k) = \Q(0) \ , \ H^2(P_k) = \Q(-1) \oplus \Q(-1) \ , \ H^4(P_k) = \Q(-2)\ . $$
The group $\mathrm{Gal}(k/\Q)$ acts trivially.  The divisor $\widetilde{A}_k$ is normal crossing in   the smooth proper scheme $P_k$, 
and  consists of 4 lines $L_1,\ldots, L_4$ meeting at 5 points given by the inverse images under $\pi$ of the four points \eqref{4points}  and  the point  $(1:0:0)$. For any subset $I\subset \{1,\ldots, 4\}$ let $L_I= \cap_{i\in I} L_i$. We set $L_{\emptyset} = P_k$.  A Gysin (residue) spectral sequence \cite{delignehodge2}  in the category $\mathcal{H}_k$ 
 has $E^{-p,q}_1= \bigoplus_{|I|=p}  H^{q-2p} (L_I)(-p)$ and converges to  $\gr^W_q H^{q-p}(P\backslash L)$. The differentials are the alternating sums of Gysin morphisms. Writing this out:
 $$\begin{array}{ccccc}
 \Q(-2)^{\oplus 5}  & \rightarrow     & \Q(-2)^{\oplus 4} & \rightarrow & \Q(-2)   \\
    &   & &&    \\
    &   & \Q(-1)^{\oplus 4} & \rightarrow & \Q(-1)^{\oplus 2}    \\
      &   &   && \\
      &   &   && \Q(0)  
\end{array}
$$
All zero entries have been omitted and in particular all rows with odd degrees are zero. 
The right-most column is the cohomology of $P_k$, the middle column the  cohomology of the union of one-dimensional strata $L_i$  with degrees shifted by 2, and the left-most column that of the  five points which constitute the codimension two strata in $\widetilde{A}_k$, with degrees shifted by $4$. 
The kernel of the map in the second row computes $\gr^W_2 H^1(P_k \backslash \widetilde{A}_k)$, which has rank $2$. 
This is because the middle row of the  right-most column is $H^2(P_k)$ which is generated by the fundamental class of a generic hyperplane (say $y=z$) and the exceptional divisor. Since the hyperplane $x=\rho z$ meets the exceptional divisor, the map in the second row is surjective and its kernel has rank 2. 
By proceeding in this way, or noting that $P_k \backslash \widetilde{A}_k$ is affine (which implies that its cohomology vanishes in degrees $\geq 3$, and so the top row of the previous diagram has all cohomology concentrated in the left-most column) 
we conclude that
$$H^0(P_k \backslash \widetilde{A}_k) =\Q(0) \ , \  H^1(P_k \backslash \widetilde{A}_k) = \Q(-1)^{\oplus 2} \ , \ H^2(P_k \backslash \widetilde{A}_k) =\Q(-2)^{\oplus 2} \ ,$$
as objects of  $\mathcal{H}_k$. The  group  $\mathrm{Gal}(k/\Q)$ permutes the two lines $x= \rho z$ and $x= \overline{\rho} z$, from which one deduces its action on $H^1$ and $H^2$ and gives the stated formula. 
\end{proof}
In algebraic de Rham cohomology, $H_{dR}^1(P \backslash \widetilde{A})$ is generated by the classes of the pullbacks under $\pi^*$ of the logarithmic one-forms
\begin{equation} \label{omegaydef}  \omega_y = d\log (y+z) - d \log (y-z) \ ,\end{equation} 
where  $\pi: P \rightarrow \Pro^2$ is the blow-up,
and  
\begin{equation} 
\label{omegaxdef} \omega_x  =   \frac{z dx - x dz}{x^2-xz+z^2} = \frac{1}{\rho- \overline{\rho}} \left( d\log (x-  \rho z) -  d\log (x- \overline{\rho} z) \right)\ , 
\end{equation} 
where $[\pi^* \omega_x]$
 generates $\Q_{\chi}(-1)_{dR}$ 
and  $[\pi^* \omega_y]$ generates $\Q(-1)_{dR}$ in $H^1(P \backslash \widetilde{A})$.  
  The class of 
 $[ \pi^*(\omega_y \wedge \omega_x)]$ generates the copy of $\Q_{\chi}(-2)_{dR}$ in $H^2(P \backslash \widetilde{A})$.

 \subsection{Face maps}
  The divisor $\widetilde{B}_k\backslash \widetilde{A}_k$ is simple normal crossing  with three smooth  components which  are the strict transforms of:
 $$(D_1)_k=  \mathcal{E}'_k= \mathcal{E}_k \backslash \Sigma$$
 where $\Sigma$ was defined in \eqref{SigmaDef}; 
  and
 $$(D_2)_k= \Pro_k^1 \backslash \{ (2:1:1)  , (2:-1:1)   \}\ ,$$
 which corresponds to the line $\{x=2z\}$ which meets $A$ along $y=\pm z$; together with 
 the  inverse image of the  exceptional divisor, which   is isomorphic to $$(D_3)_k=\Pro_k^1 \backslash \{\rho, \overline{\rho}\}\ .$$ They are the extension of scalars to $k$ of three divisors $D_1,D_2,D_3$ over $\Q$ which meet each other as depicted in figure 2 (for example, $D_1$ is the strict transform of $\mathcal{E} \backslash (\mathcal{E} \cap \mathcal{A})$, and $D_3= \Pro^1 \backslash V(x^2-xz+z^2)$.)
  
 For each  $i=1,2,3$  there are `face maps' \cite{brownnotesmot}, \S10.3: 
\begin{equation}\label{facemapi}  f_i :  H^1(D_i, D_{ij} \cup D_{ik}) \To M
\end{equation}
where $\{i,j,k\}=\{1,2,3\}$ and $D_{pq} = D_p \cap D_q$.

 \begin{lem} Since $D_1,D_2, D_3$ are  over $\Q$, they define the following  objects of $\mathcal{H}$:
 $$H^1(D_1) \cong H^1(\mathcal{E}) \oplus \Q^{\oplus 3} (-1) \oplus \Q_{\chi}(-1)\ ,$$ 
 whereas 
$H^1(D_2) = \Q(-1)$ and $H^1(D_3) = \Q_{\chi}(-1)$. Their respective $H^2$'s all vanish. 
 \end{lem}
 \begin{proof}  We first work in the category $\mathcal{H}_k$. 
 The cohomology of $(D_1)_k$ is given by a Gysin (residue) sequence:
 $$ 0 \To H^1(\mathcal{E}_k)  \To H^1((D_1)_k) \To   \widetilde{H}_0(\Sigma)(-1) \cong \Q(-1)^{\oplus 4}  \To 0    $$
 where the $\Q(-1)$ on the right are objects of $\mathcal{H}_k$.  The previous sequence splits by   the Manin-Drinfeld theorem, since the points removed from $\mathcal{E}_k$ are cusps (a splitting is provided by the action of Hecke operators).  The first statement follows
 since $\mathrm{Gal}(k/\Q)$  acts trivially on  the points  $(0:1:0)$, $(0:\pm 1:1)$,  but permutes $ (\rho:0:1)$ and  $( \overline{\rho}:0:1)$. Since $D_1,D_2,D_3$ are affine, their cohomology in degrees 2 and above vanish. The second statement follows from $D_2 \cong  \G_m$, and 
 $(D_3)_k= \Pro^1 \backslash \{\rho, \overline{\rho}\}$. 
 \end{proof} 

 \subsubsection{Computation of $M$}  We use the well-known relative cohomology spectral sequence in the category $\mathcal{H}$. It satisfies:
   $$E_1^{p,q} = \bigoplus_{|I|=p}  H^q(D_I)    \quad \Longrightarrow \quad  H^{p+q} \left(P \backslash \widetilde{A},   \widetilde{B} \backslash \left(  \widetilde{B} \cap \widetilde{A}\right) \right)     $$
   where $D_i$ for $i\in \{1,2,3\}$ denote the affine schemes  above, $D_{\emptyset} =  P \backslash \widetilde{A}$ and for every non-empty  subset $I \subset \{1,2,3\}$, we write $D_I = \cap_{i \in I} D_i$. The differentials are given by signed sums of restriction maps. 
\begin{prop}  The weight-graded pieces of $M$ are:
\begin{equation} \label{weightgradedM} \gr^W M =  \Q_{\chi}(-2)\oplus  \Q(-2) \oplus \Q_{\chi}(-1) \oplus \Q(-1)^{\oplus 3} \oplus H^1(\mathcal{E})  \oplus \Q(0)^{\oplus 2}\ .
\end{equation}
More precisely, we have $M=W_4M$, $\gr^W_3 M=0$,  and 
$$ M/W_2 M  \cong \gr^W_4 H^2  (P \backslash \widetilde{A}) =   \Q_{\chi}(-2)\oplus  \Q(-2) \ .$$
Its  weight two part splits into a direct sum
\begin{equation} \label{W2Msplit} W_2 M  \cong H^1(\mathcal{E}) \oplus T\ ,
\end{equation}
where $T$ is an extension:
$$0 \To \Q(0)^{\oplus 2} \To T \To \Q(-1)^{\oplus 3} \oplus \Q_{\chi}(-1) \To 0\ .$$
  \end{prop}
  
  \begin{proof}  We can work in $\mathcal{H}$. 
   The first page of the spectral sequence is
    $$\begin{array}{ccccc}
  \Q_{\chi}(-2)   
\oplus \Q(-2)    & \rightarrow     & 0  & \rightarrow &  0  \\
\Q_{\chi}(-1)  \oplus   \Q(-1)   & \rightarrow  & H^1(\mathcal{E}) \oplus \Q_{\chi}(-1)^{\oplus 2} \oplus \Q(-1)^{\oplus 4} &  \rightarrow & 0   \\
   \Q(0)    & \rightarrow  &  \Q(0)^{\oplus 3}   & \rightarrow & \Q(0)^{\oplus 4}  
\end{array}
\ .$$
The column on the far left is given by lemma \ref{lem: prelim}. The column on the far right is the cohomology of the union of the pairwise intersections $D_i \cap D_j $, which consists of 4 points. 
The structure \eqref{weightgradedM} follows from the fact that the left-most differential in the middle row is injective (for example, one can check  that the classes $[\pi^* \omega_x]$ and $[\pi^* \omega_y]$ restrict to non-trivial classes in the de Rham cohomology of $D_3$, and $D_2$ respectively).  
  Next, by taking the quotient by $W_2$ in the  natural map $M \rightarrow H^2  (P \backslash \widetilde{A})$ we obtain the second statement.   Now, since $P\backslash \mathcal{A}$ is affine, we know by \cite{brownnotesmot} proposition 10.7 (or by inspection of the spectral sequence above) that the sum of the face maps is surjective. In other words, the map
  $$ \sum_i f_i:  \bigoplus_{i=1}^3 H^1(D_i, D_{ij} \cup D_{ik}) \To W_2 M$$
  is surjective, where, in the above sum, $j,k$ are chosen such that $\{i,j,k\}=\{1,2,3\}$. In particular $W_2M$ is a quotient of this direct sum. To obtain the splitting \eqref{W2Msplit}, it suffices to show that $H^1(\mathcal{E})$ is a summand of:
  $$H^1(D_1 , D_{12} \cup D_{13})\ ,$$
which follows again from the  Manin-Drinfeld theorem since $D_{12}, D_{13}$ correspond to cusps on $\mathcal{E}$.   One can also prove this fact by direct application of Hecke operators to $H^1(D_1 , D_{12} \cup D_{13})$.  
 It follows from this that $H^1(\mathcal{E})$ is a summand in $W_2 M$. 
  
  We conclude that $W_2M = H^1(\mathcal{E}) \oplus T$ for some object $T$ of $\mathcal{H}$ whose weight-graded pieces are Tate or of the form $\Q_{\chi}(-1)$.  That it is an extension of the stated form follows from \eqref{weightgradedM}.
  \end{proof} 
 Consider  the exact sequence 
  $$0 \To W_2 M \To M \To \Q(-2) \oplus \Q_{\chi}(-2)\ .$$
  We can pull it  back  to a simpler object $N\rightarrow M$, which sits in an exact sequence
\begin{equation} \label{Ndefn}
0 \To W_2 M = W_2N \To N  \To \Q_{\chi}(-2)\To 0 \ .
\end{equation} 
  By \eqref{W2Msplit}, this extension can in turn  be pushed out  to a simple extension:
  \begin{equation} \label{N1ext}  0 \To H^1(\mathcal{E}) \To N_1 \To \Q_{\chi}(-2) \To 0
  \end{equation} 
  and a biextension of the form 
  \begin{equation} \label{N2ext}  0 \To T \To N_2 \To \Q_{\chi}(-2) \To 0 \ .
  \end{equation} 
  The Hodge numbers of $N_2$ are of Tate type.
  %
  %and one can presumably prove without  difficulty (by replacing $\mathcal{E}$ with a suitable configuration of lines which intersect all the others everywhere that $\mathcal{E}$ does) that it is an object in the category of   mixed Artin-Tate motives of \cite{DG}. \claudecomment{Missing reference.}  
  \subsection{The motivic periods}  Consider the form
  \begin{equation} \label{omegadef} \omega = \omega_x \wedge \omega_y =   \frac{1}{\rho- \overline{\rho}} \,  d\log \left(\frac{x-\rho z}{x- \overline{\rho}
  z} \right)\wedge d \log \left(\frac{y+z}{y-z}\right)
  \end{equation}
  Its restriction to the affine chart $z=1$ is
  $$ \omega\big|_{z=1} = \frac{2\, dx \wedge  dy }{(x^2-x+1)(1-y^2)} \ .$$ 
  It defines a cohomology class $[\pi^* \omega] \in F^2 M_{dR}$ whose image in $H^2(P \backslash \widetilde{A})$ spans the copy of $\Q_{\chi}(-2)_{dR}$. 
 Given any relative homology class $[\sigma] \in M_B^{\vee}$, we can consider the motivic period defined by the matrix coefficient (\cite{brownnotesmot} \S2): 
 $$ \xi= [ M, [\sigma], [\pi^* \omega]]^{\mm} \ .$$
 Its image under the period homomorphism is the period
 $$\mathrm{per}  \, \xi=   \int_{\sigma} \omega \ .$$
 If the class $\sigma$ is  invariant  (resp. anti-invariant) under $F_{\infty}$ then the associated period is real (imaginary). We shall mainly consider two examples of real periods. 
 
 \section{Relative homology classes in $M_{B}^{\vee}$}
 \label{sec:Sec4}
  \subsection{Frobenius-invariant chains} Recall that the Betti component $M_B$ of an object $M$ in $\mathcal{H}$ comes with an action of the real Frobenius $F_{\infty}$. 
  Since it acts on $H_B^1(\mathcal{E})$ with eigenvalues $+1$ and $-1$, 
  it follows from \eqref{weightgradedM} and the definition of $N$ that the $+$ eigenspace for the action of  $F_{\infty}$  on $N_B$ has dimension $4$. Each eigenspace comes from the Betti component of a weight-graded piece $\Q(0)^{\oplus 2}, \Q_{\chi}(-1)$ and $H^1(\mathcal{E})$ in \eqref{weightgradedM}.    
In particular,  since $F_{\infty}$ acts via $-1$ on $\Q_{\chi}(-2)_B$, we have
 \begin{equation} \label{NBplussplit} 
 N_B^+ = W_2 N_B^+ = W_2 M_B^+
 \overset{\eqref{W2Msplit}}{=}
  H^1(\mathcal{E})^+_B \oplus  T_B^+\ .
  \end{equation} 
  Here, a superscript $\pm$ refers to corresponding eigenspace under $F_{\infty}$. 
  Let $[\sigma_{\mathcal{E}}] \in (N_B^+)^{\vee}$ denote the image of a generator of  the Frobenius-invariant part of the singular homology 
  $H_1(\mathcal{E}(\C))^+$ of the elliptic curve $\mathcal{E}(\C)$. 
  
  We first discuss how to obtain relative homology classes from paths, before writing down  a representative for the class $[\sigma_{\E}]$ explicitly. 
  
  \subsection{Paths and  relative homology classes}
  The periods we wish to consider  are iterated integrals of logarithmic one-forms along  paths in $\mathcal{E}(\C)$. We now explain how these paths define relative homology classes in  $M_B^{\vee}$.  
  
 The  projection $\pi_x$ extends to a double covering $\mathcal{E}\rightarrow \Pro^1$  by sending $(0:1:0)$ to the point at infinity. It is     ramified  at $\infty$ and cube roots of $-1$.
  The image of  $\Sigma$ \eqref{SigmaDef} are the points 
  $\{0, \rho, \overline{\rho}, \infty\}$. Consider any continuous  path 
   $$\gamma:(0,1)  \rightarrow \Pro^1(\C) \backslash \{0,-1,\rho, \overline{\rho}, \infty\}\ , $$  which extends to a continuous path $\gamma:[0,1] \rightarrow \Pro^1(\C)$ with the property that $\gamma(0) \in \{2, \infty\}$ and $\gamma(1)=\infty$.  Such a path,  together with the data of a determination of $\sqrt{x^3+1}$ at any point  $\gamma(t)$ for $0< t < 1$ defines a path on $\mathcal{E}(\C)$ whose endpoints are contained in the set  $\{(2:3:1), (2:-3:1), \infty\}$. The latter are the points of intersection of $\mathcal{E}$ and $x=2z$ (i.e., the dimension 0 strata of the divisor $B$). 
   
  \subsubsection{Chains constructed from  paths} 
   Given $\gamma$ as above, 
    consider         
    the singular 2-chain $p(\gamma)$ defined by the   map
       \begin{eqnarray} \{0< t_1 < t_2 < 1\} & \To &   \Pro^2(\C) \nonumber \\
       (t_1,t_2)  & \mapsto &  \left(\gamma(t_1) : \sqrt{\gamma(t_2)^3+1} :1\right)  \nonumber 
       \end{eqnarray}  
        where the determination of the square root is uniquely determined from the defining data by analytic continuation along  $\gamma$.  Denote by $\widetilde{p}(\gamma) =  \overline{\pi^{-1} (p(\gamma))}$ 
         the closure in the analytic topology of the inverse image of $p(\gamma)$ under $\pi:P(\C)\rightarrow \Pro^2(\C)$. 
         Since $\gamma$ avoids $0$ and  the three cube roots of $-1$, it follows that $\widetilde{p}(\gamma)$
         does not meet $A(\C)$.  Its boundary  is contained in the locus $B(\C)$
       by assumption on the  endpoints of $\gamma$:  the boundary component corresponding to $t_1=t_2$ is contained in the elliptic curve  $\mathcal{E}(\C)$, and those corresponding to  $t_1=0$ and $t_2=1$ are contained in the exceptional divisor or the inverse image $x=2z$.   Thus we have shown:
       \begin{lem} The chain $\widetilde{p}(\gamma) $   defines a relative homology class 
       $$ [ \widetilde{p}(\gamma)] \quad \in \quad M_B^{\vee} \ .$$
       \end{lem} 
       Consider the following examples:

       \begin{enumerate} 
       \item 
The straight-line path $\gamma_{2,\infty}$ from $2$ to $\infty$ which is contained in the real axis, together with the positive root of $x^3+1$. The  chain
$ p(\gamma_{2,\infty})$ 
is
  $$ \{ (x:y:1):  2 < x  \ , \  3< y  \  , \   x^3+1 <y^2 \} \quad   \subset \quad  \Pro^2(\R)\ .$$
 The closure of its inverse image in $P(\R)$ defines a relative homology cycle whose class $[ \widetilde{p}(\gamma_{2,\infty})]\in M_B^{\vee} $  
  which is invariant under $F_{\infty}$.

\vspace{0.1in}
\item
Let $\gamma_{\infty,-1}$ denote a path (together with the positive square root of $x^3+1$ initially) which travels along the real axis  from $\infty$  to a point close to $0$ around which it traverses  in a small semi-circle, before continuing on to  a  point near $-1$ along the real axis. 
  After winding around $-1$, it  returns back towards infinity, this time passing around $0$ on the opposite side. The sign of $\sqrt{x^3+1}$ is  negative on the return path.   Let $\overline{\gamma}_{\infty,-1}$ denote the complex conjugate path (but equipped with the same, initially positive, determination of the square root of $x^3+1$). 
The linear combination 
 $$\frac{1}{2} \left( [\widetilde{p}(\gamma_{\infty, -1})]  + [\widetilde{p}(\overline{\gamma}_{\infty,-1})] \right)   $$
 is invariant under $F_{\infty}$. It is not zero because we are working with paths in the elliptic curve, or, `loaded' paths on the punctured sphere with coordinate $x$. 
      \end{enumerate}
  There are many  other paths which one might consider, including paths from $x=2$ to $\infty$ which wind around the singularities $0$, or $\rho$ and $\overline{\rho}$.

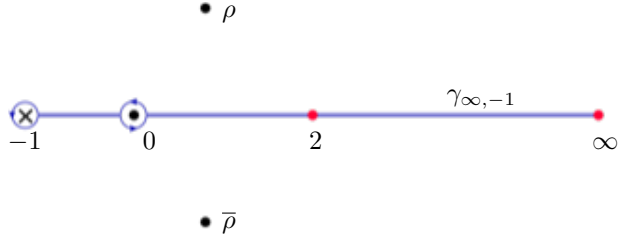
\begin{figure}[h!]
{\includegraphics[height=3cm]{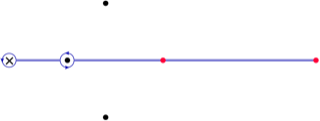}} 
\put(-5,30){$\infty$}
\put(-175,30){$0$}
\put(-226,30){$-1$}
\put(-145,80){$\rho$}
\put(-145,0){$\overline{\rho}$}
\put(-112,30){$2$}
\put(-60,48){$\gamma_{\infty,-1}$}
\caption{The path $\gamma_{\infty,-1}$  in $\C \backslash \{0,-1, \rho, \overline{\rho}\}$ relative to the two points $\{2, \infty\}$ (red).  The  punctured elliptic curve is a double cover, ramified at the additional point $-1$.}
\end{figure}
\noindent

     \subsection{The elliptic extension}  \label{sec: EllExt} 
The class $[\sigma_{\mathcal{E}}]$ can be represented as follows. Start with the real locus $\mathcal{E}(\R)$ (oriented in the positive $y$ direction)  and deform it by small semi-circles around the points $(0:\pm 1: 1)$ so that it avoids $A \cap \mathcal{E}$ as shown in figure 4; the upper line depicts its image under $\pi_y$.   The resulting chain  $c$ is not invariant under complex conjugation, but  $\frac{1}{2} (c+ \overline{c})$ 
is  a representative  for a  Frobenius invariant path in $H_1(\mathcal{E} \backslash (\mathcal{E} \cap A) (\C) )$.  
It can be viewed as  the path  given by the real locus $\mathcal{E}(\R)$ away from $(0: \pm 1:1)$ which  bifurcates into two `half-paths' near each point $(0:\pm 1:1)$ - each half-path traces a semi-circle on either side of the puncture which meet on the other side. We can view $\mathcal{E}(\C) \subset P(\C)$. 
By lemma \ref{lem: prelim},  the images of $c$ and $\overline{c}$  vanish in  $H_1(P \backslash \widetilde{A})$ (for instance, the integrals of $\pi^*\omega_y$, $\pi^*\omega_x$ vanish along them) and so 
 there exists a singular two-chain $\sigma_{\mathcal{E}} $ in $ (P \backslash \widetilde{A})(\C)$ such that  $\partial \sigma_{\mathcal{E}}=\frac{1}{2}( c+\overline{c})$.  Since the   boundary of $\sigma_{\varepsilon}$ is contained in the divisor $B(\C)$,  it defines a relative homology class  in $M_B^{\vee}$.

\begin{figure}[h!]
{\includegraphics[height=0.88cm]{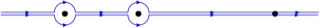}} 
\put(-50,-10){$y=\infty$}
\put(-170,-10){$y=1$}
\put(-240,-10){$y=-1$}
\caption{A singular chain  in $\A^1\backslash \{\pm 1\}$ given by the sum of the classes of the two paths shown. 
It is the image of a Frobenius-invariant path on $\mathcal{E}(\C)$ under the projection $\pi_y: \mathcal{E}\rightarrow \Pro^1$. 
}
\end{figure}

The integral of $\pi^*(\omega)$ \eqref{omegadef} along  $\sigma_{\E}$ can be computed as follows. Consider the  primitive 
$$ F =   -\omega_x \,  \log \left(\frac{y+z}{y-z } \right)  $$
of $\omega$. It satisfies $dF = \omega$. By Stokes' formula, and the fact that $\partial \sigma_{\varepsilon}=\frac{1}{2}(c+\overline{c})$, 
$$\int_{\sigma_{\E}} \pi^* (\omega) = \int_{\frac{1}{2}(c+\overline{c})} F\Big|_{\mathcal{E}} =  \mathrm{Re}\, \int_{c} F\Big|_{\mathcal{E}}  \ .$$
The last part follows from the fact that the chain of integration is Frobenius-invariant and hence  the integral is real. 
The integrand, in the coordinate $y$,  is
  $$  F\Big|_{\mathcal{E}}=- \frac{2}{3} \left(\frac{x+1}{x^2}\right) \log \left( \frac{y+1}{y-1}\right) \frac{dy}{y}$$
  where $x = \sqrt[3]{y^2-1}$ is the branch given by the real root for $y$ large on the real axis.

      \subsubsection{Reformulation} 
  Since the real part of the integrand is anti-invariant under the involution $y\mapsto -y$, it suffices to integrate along the segment of $c$ from $y=0$ to infinity. The point $y=0$ corresponds to the point $x=-1$, which does not play any role in the definition of the motive $M$, but this does not matter. Writing the previous integral  using the $x$ coordinate gives
  \begin{eqnarray}  \label{Iinpieces}    
  \int_{\sigma_{\E}} \pi^* (\omega)  &=  & -2\,  \mathrm{Re}\, \int^{\infty}_{-1}   \frac{dx}{x^2-x+1}  \log \left( \frac{ \sqrt{x^3+1}+1}{  \sqrt{x^3+1}-1  }  \right)  \nonumber \\
    & =  &  -2 \int_{-1}^0   \frac{dx}{x^2-x+1}  \log \left( \frac{ 1+ \sqrt{x^3+1}}{ 1- \sqrt{x^3+1}  }  \right)  + \\
 &  & \quad-   2  \int_{0}^{\infty}  \frac{dx}{x^2-x+1}  \log \left( \frac{ \sqrt{x^3+1}+1}{  \sqrt{x^3+1} -1 }  \right)    \nonumber  \\
 & < & 0    \ . \nonumber 
 \end{eqnarray}
Each integral converges (since it is an integral on a compact domain with boundary, and has at worst logarithmic singularities on the boundary: see \cite{Brown:2009qja}, \S4.4) and is negative. 
The image of the  path $c$ under the projection $\pi_x$ is equivalent to  the path $\overline{\gamma}_{\infty,-1}$ depicted in  figure $3$.

\begin{rem}  \label{remMMVinterp} Each integral in \eqref{Iinpieces}   can be interpreted as a multiple modular value \S\ref{sec:modular}, since they are  regularised double integrals of modular forms between cusps.
They can also be interpreted as  multiple elliptic polylogarithms: i.e.,  an iterated integral of the  two logarithmic forms $\omega_x$, $\omega_y$  along a path in $(\mathcal{E}\backslash \Sigma)(\R)$ between tangential basepoints based at $\Sigma$.     \end{rem} 
   
   \section{Two motivic periods which are not polylogarithmic}
   \label{sec:two_periods}
   
     \subsection{A mixed-elliptic period} 
     \label{sec:mixed_elliptic_period}
     Consider the motivic period 
     $$I^{\mm}_{\mathcal{E}}:= [N, [\sigma_{\mathcal{E}}],  [\pi^*\omega]]^{\mm}   \ . $$
     It is equivalent, via the morphism $r: N\rightarrow N_1$ (see \eqref{N1ext}) to the  motivic period
     $$  [ N_1, [\sigma_{\mathcal{E}}], r_{\mathrm{dR}}\, [\pi^*\omega]]^{\mm}$$ 
      of $N_1$, since 
$[\sigma_{\mathcal{E}}]$ is by construction in the image of $r_B^{\vee}: (N^{\vee}_1)_B \rightarrow N^{\vee}_B$ (denoted with the same symbol).  Its period is therefore an $F_{\infty}$-invariant period of the simple extension $N_1$. We have just shown  that it is negative, and hence  non-zero:
$$\mathrm{per} \left( I^{\mm}_{\mathcal{E}}\right)= I_{\mathcal{E}} <0 \ .$$  
 We now show that this  is precisely the obstruction to being a polylogarithmic motivic period.
 
 \subsubsection{Non-triviality of the extension}
\begin{lem} The extension $N_1$ does not split. 
\end{lem} 
\begin{proof} If one had $N_1 \cong H^1(\mathcal{E}) \oplus \Q_{\chi}(-2)$, i.e., $N_1$ were to split in $\mathcal{H}$, then  $I^{\mm}_{\mathcal{E}}$ would be a  sum  $I_1^{\mm} + I_2^{\mm}$  where $I_1^{\mm}$, $I^{\mm}_2$ are  motivic periods of $H^1(\mathcal{E})$, $\Q_{\chi}(-2)$, respectively. Because $I^{\mm}_{\mathcal{E}}$ is real ($F_{\infty}$-invariant), we can assume that the same is true of both $I_1^{\mm}, I_2^{\mm}$.   But since $F_{\infty}$ acts via $-1$ on $\Q_{\chi}(-2)_B$, the object $\Q_{\chi}(-2)$ has no non-trivial real periods and so $I_2^{\mm}=0$. Therefore  $I^{\mm} = I_1^{\mm}$ is a motivic period of $H^1(\mathcal{E})$. Furthermore,  it has Hodge filtration $F^2$,  but since $F^2 H_{dR}^1(\mathcal{E})  =0$, we must have $I^{\mm}=0$. Therefore  if $N_1$ were to split then $I^{\mm}_{\mathcal{E}}$ would vanish 
 and so would its period 
 $I_{\mathcal{E}} = \mathrm{per} \left( I_{\mathcal{E}}^{\mm}\right)$, a contradiction. 
\end{proof} 

\begin{cor}\label{cor:main_1} The motivic period $I^{\mm}_{\mathcal{E}}$ is algebraically independent over the  motivic periods of  mixed Artin-Tate objects in $\mathcal{H}$.  In particular, it is not  equivalent to a polylogarithmic motivic period. 
\end{cor} 
\begin{proof} Since the extension $N_1$ is non-split, the unipotent radical of the de Rham Galois group $G^{dR}_{\mathcal{H}}=\mathrm{Aut}_{\mathcal{H}}^{\otimes}{\omega_{dR}}$, where $\omega_{dR}$ is the de Rham fiber functor,  acts non-trivially on its de Rham realisation and also on the motivic period $I^{\mm}_{\mathcal{E}}$.  It therefore admits a Galois conjugate $\xi = (g-\id)  I^{\mm}_{\mathcal{E}}$  for some $g\in G^{dR}_{\mathcal{H}}(\overline{\Q})$, where $\xi$ is a  non-zero motivic period  of $H^1(\mathcal{E})$.  We may write $\xi = \alpha \,\omega_+ + \beta\, \eta_+$ where $\alpha,\beta \in \Q$ are not both zero and $\omega_+, \eta_+$ are real (i.e., $F_{\infty}$-invariant) motivic periods of  $H^1(\mathcal{E})$  of Hodge types $(1,0)$ and $(0,1)$  respectively (see  \cite{brownnotesmot} for definitions). Since Artin-Tate objects in $\mathcal{H}$ are all of Hodge type $(p,p)$, the element $\xi$ is algebraically independent over the ring $\Pe^{\mm}_{AT}$ generated by  their motivic periods.  Suppose by contradiction that $I^{\mm}_{\mathcal{E}}$ is algebraic over $\Pe^{\mm}_{AT}$, and thus satisfies an equation $P(I^{\mm}_{\mathcal{E}})=0$ where $P =a_n x^n +\ldots +a_0  \in \Pe^{\mm}_{AT}[x] $ is a polynomial  with coefficients in $\Pe^{\mm}_{AT}$ and $a_n \neq 0 $. Since the ring $\Pe^{\mm}_{AT}$ is stable under 
 $G^{dR}_{\mathcal{H}}$, we can apply $g$ to $P(I_{\mathcal{E}}^{\mm})$ to deduce a non-trivial polynomial equation for $\xi$ of the form $g(a_n) \xi^n + \ldots =0 $ whose coefficients are in   $\Pe^{\mm}_{AT}[I^{\mm}_{\mathcal{E}}]$,  since  $g(a_n)\neq 0$. Since a composition of algebraic extensions is algebraic, it follows that $\xi$ is algebraic over $\Pe^{\mm}_{AT}$, a contradiction.
\end{proof} 
In fact, the period $I_{\mathcal{E}}$ is  proportional to the regulator of the extension.  By Beilinson's conjecture it is predicted to be a special value of the $L$-function of the elliptic curve at $2$ and indeed we find numerically to many digits that 
\begin{equation} 
\label{IellipticNumerically}  I_{\mathcal{E}}  \overset{\cdot}{=} -4 \, \pi \sqrt{3}\,  \Lambda(f,2)\ .   
\end{equation}
In \S\ref{sec:modular} we discuss a way of computing the left-hand side to high precision.

 \subsection{A mixed Artin-Tate-elliptic period}
 \label{sec:I_pol}
 Let us now consider  the locus 
  $$\sigma = \{ (x:y:1):  x> 2 , 3<y  ,   x^3+1 <y^2 \}  \ \subset \  \Pro^2(\R)$$
  which is the chain $p(\gamma_{2,\infty})$ considered earlier. The closure $\widetilde{p}(\gamma_{2, \infty})$ of its pull-back  to $P(\C)$  defines a relative homology cycle whose class 
  $ [\widetilde{p}(\gamma_{2, \infty})] \in M_B^{\vee}$ is invariant under $F_{\infty}$.  We shall denote it  simply by $[\sigma_M]$. 
Consider the  motivic period 
  $$I^{\mm} = [ M, [\sigma_M] , [\pi^*\omega]]^{\mm}\ .$$
  Its period is given by the following integral along the path $\gamma_{2, \infty}$:
  \beq\label{eq:I_def_2}
  \per\, \left( I^{\mm} \right)=   \frac{1}{\rho- \overline{\rho}}  \int_{2\leq x_1 \leq x_2 \leq \infty}   
  d\log \left(\frac{x_1-\rho }{x_1- \overline{\rho}
  } \right)\wedge d \log \left(\frac{\sqrt{x_2^3+1}+1}{\sqrt{x_2^3+1}-1}\right)
  \eeq
   where the square roots are positive.     Since the class of $[\pi^*\omega]$ spans the copy of $\Q_{\chi}(-2)_{dR}$ in $\gr^W_4M_{dR}$, the natural map $N\rightarrow M$ (see \eqref{Ndefn}) defines an equivalence of motivic periods 
    $$I^{\mm} = [ N, [\sigma_M] , [\pi^*\omega]]^{\mm}\ ,$$
    where, by abuse of notation, $[\sigma_M] $ also denotes its image in $M_B^{\vee} \rightarrow N_B^{\vee}.$ By \eqref{NBplussplit}, 
 there exists a rational number $\lambda_{\mathcal{E}} \in \Q$ such that
 $$[\sigma_M] = \lambda_{\mathcal{E}} [\sigma_{\mathcal{E}}] +  [\sigma_{T}] \  ,$$
 where $[\sigma_{\mathcal{E}}]$ is the elliptic class considered earlier, and $[\sigma_{T}] \in T^{\vee}_B$ is some relative homology class in the Artin-Tate object $T$.
 It follows that $I^{\mm}$ is a sum:
 \beq\label{eq:relation}
I^{\mm} = \lambda_{\mathcal{E}}  I^{\mm}_{\mathcal{E}} +   I^{\mm}_{\mathrm{Pol}} \,,
\eeq
where
$$ I^{\mm}_{\mathrm{Pol}}   =   [ T, [\sigma_{T}] , [\pi^*\omega]]^{\mm}  \nonumber $$
is a period of an Artin-Tate motive $T$.
One can presumably show that $I^{\mm}_{\mathrm{Pol}}$ is a linear combination of  motivic dilogarithms and logarithms, as the notation suggests, although we have not done this.

 It remains  to compute the coefficient $\lambda_{\mathcal{E}}$. The boundary  component of  $[\sigma_M]$  which lies in $\mathcal{E}(\C) \backslash \Sigma$ (i.e., its image under the dual  of the Betti component of the face map \eqref{facemapi} for $i=1$) is the path $\alpha$ from $(2,3)$ to $\infty$ in $\mathcal{E}(\C)$. We can check (by using the relations obtained by intersecting $\mathcal{E}$ with the lines $x=2z$, $x+z=y$ and $y=z-2x$) that the orbit of the point at infinity $\infty$ under  multiplication by $(2,-3)$ in the group law of the elliptic curve is: 
 $$ \infty  \ \mapsto  \ (2,-3) \  \mapsto \  (0,-1)  \ \mapsto  \ (-1,0) \  \mapsto \  (0,1)  \ \mapsto  \ (2,3) \ \mapsto \ \infty \ .$$
 It follows that  $6 \alpha $ is homotopic to the  path from $y=-\infty$ to $y=\infty$, which is  the  Frobenius-invariant homology generator on $\mathcal{E}(\C)$ considered in \S\ref{sec: EllExt}. Therefore, 
 $$\lambda_{\mathcal{E}} = \frac{1}{6} \ . $$ 
In particular, $\lambda_{\mathcal{E}}$ is non-zero, and since $I^{\mm}_{\mathrm{Pol}}$ is Artin-Tate,  we deduce the
\begin{cor}\label{cor:main_2} The motivic period $I^{\mm}$ is algebraically independent over the space of mixed Artin-Tate motivic periods. In particular, it is algebraically independent from motivic polylogarithms at algebraic points. 
\end{cor} 
The period conjecture, in the weak version stated in \cite{brownnotesmot}, Conjecture 1, implies that the period homomorphism from the ring of motivic periods of $\mathcal{H}$  to $\C$ is injective. If true, as expected, then it implies that the integrals $I$ and $I_{\mathcal{E}}$ are algebraically independent from values of polylogarithms at algebraic arguments.

\begin{rem}
  The above discussion involved no numerical or analytic calculations, only the negativity of the integral $I_{\mathcal{E}}$ to exhibit a non-trivial extension class. In general, the underlying geometry, via the theory of motivic periods,  enables one in principle to predict completely  the types of numbers one expects to obtain. 
\end{rem}

%% file: modular.tex
% !TEX root = main.tex

\section{Double Eisenstein integrals and $L$-values of cusp forms}
\label{sec:modular}

\subsection{Eisenstein series on $Y(6)$}
 Every Eisenstein series of weight $n \geq 2$ for $\Gamma(6)$ is a linear combination of the following series~\cite{Broedel:2018iwv}:
\beq\label{eq:EisensteinH_def}
H^{(n)}_{r,s}(\tau) = \sum_{\substack{(\alpha,\beta)\in \mathbb{Z}^2\\ (\alpha,\beta)\neq(0,0)}}\frac{e^{i\pi(s\alpha-r\beta)/3}}{(\alpha+\beta\tau)^n}\,,\qquad 0\le r,s<6\,.
\eeq
This series is absolutely convergent, unless $n=2$, in which case the `Eisenstein summation' convention is understood. In Appendix~\ref{app:Eisenstein} we show how to express differentials with logarithmic singularities at the cusps in terms of these Eisenstein series and the cusp form $f$.

%%%%%%%%%%%%%%%%%%%%%%%%%%%%%%%%%%%%%
\subsection{Double Eisenstein integrals} We can use the modular parametrisation of $\E$ to write the iterated integral of~\eqref{eq:I_def_2} as an iterated integral on $Y(6)$. Changing variables from $x$ to $\tau$ using~\eqref{eq:x6y6_def} and using the relations in Appendix~\ref{app:Eisenstein}, we find
\beq\bsp\label{eq:dlog_to_eisenstein}
\varphi^*d\log&\left(\frac{x-\rho}{x-\bar\rho}\right) = \frac{d\tau}{2\pi i}\,E_1(\tau)\,,\\
\varphi^*d\log&\left(\frac{1+\sqrt{1+x^3}}{1-\sqrt{1+x^3}}\right) = \frac{d\tau}{2\pi i}\,E_2(\tau)\,,
\esp\eeq
where $E_1(\tau)$ and $E_2(\tau)$ are  the following linear combinations of Eisenstein series of weight two:
\beq\bsp
E_1(\tau)&\, = 2\,H_{1,2}^{(2)}(\tau )-H_{0,2}^{(2)}(\tau )-2\, H_{1,4}^{(2)}(\tau )-H_{2,0}^{(2)}(\tau )-2\, H_{2,2}^{(2)}(\tau )\,,\\
E_2(\tau)&\,  = 3\, H_{1,0}^{(2)}(\tau )-2\, H_{0,3}^{(2)}(\tau )+6\, H_{1,3}^{(2)}(\tau )-H_{3,0}^{(2)}(\tau )\,.
\esp\eeq
We can then recognise the integral in~\eqref{eq:I_def_2} as a double iterated integral of Eisenstein series:
\beq\label{eq:I_Eisenstein_2}
I=\frac{1}{\rho-\bar{\rho}}\int_{0\le t_1\le t_2\le\infty}\frac{dt_1\wedge dt_2}{(2\pi)^2}\,E_1(it_1)\,E_2(it_2)\,,
%I = \frac{1}{\rho-\bar{\rho}}\int_{0\le t'\le t\le \infty}\frac{dt'\wedge dt}{(2\pi)^2}\,E_1(it')\,E_2(it)\,,
\eeq
In general, 
iterated integrals of Eisenstein series may diverge at the cusps. These divergences can be regularised by replacing the cusps by a suitable a tangential base point at a cusp.  See~\cite{Brown:mmv}  for a more detailed discussion.

In~\cite{Brown:mmv} it was shown that for small weights double Eisenstein integrals for the full modular group $\Gamma(1)$ can be evaluated in terms of multiple zeta values. The first obstruction to multiple zeta values appears in weight 12. The first cusp form for $\Gamma(1)$ also appears in weight 12, and it was shown in~\cite{Brown:mmv} using the Rankin-Selberg method that certain double Eisenstein integrals in weight twelve also evaluate to the first non-critical $L$-value of this cusp form.  

In the present setting we are dealing with Eisenstein series for the subgroup $\Gamma(6)$. One  expects that in low weights double Eisenstein integrals for $\Gamma(6)$ evaluate  to multiple polylogarithms evaluated at sixth roots of unity, as well as periods of simple extensions of motives of  cusp forms, which, by Beilinson's conjecture, should include the critical values of the associated $L$-functions.  Since $Y(6)$ has genus one, the first cusp form $f(\tau)$ for $\Gamma(6)$ appears in weight two. 
%
%We therefore expect that the double Eisenstein integral in~\eqref{eq:I_Eisenstein_2} evaluates to multiple polylogarithms and the period of $f$,
% \beq
%\Lambda(f,2) = \int_{i\infty}^0d\tau\,f(\tau)\,\tau = 0.85718907492991773071685111\ldots\,.
%\eeq
%It is the completed version of the associated $L$-function is 
%\beq
%\Lambda(f,s) = \pi^{-s} \Gamma(s) 3^s  L(f,s)
%\eeq
%and satisfies the functional equation $\Lambda(f,s)= \Lambda(f,2-s)$. 
%\claudecomment{Comment on why we expect $s=2$?}

\subsection{Numerical evaluations} 
It is easy to evaluate the integral in~\eqref{eq:I_Eisenstein_2} numerically to several hundred digits. Using the PSLQ algorithm, we can find a linear combination of $\Lambda(f,2)$ and multiple polylogarithms evaluated at sixth roots of unity that agree with the numerical value of $I$ to (at least) 200 digits. We find:
\beq\label{eq:I_res}
I = -\frac{2 \, \pi}{ \sqrt{3}}\,  \Lambda(f,2) +\frac{5}{\sqrt{3}}\,\textrm{Cl}_2\left(\frac{\pi}{3}\right)\,,
\eeq
with $\textrm{Cl}_2\left(\frac{\pi}{3}\right) = \textrm{Im }\textrm{Li}_2(e^{i\pi/3})$.
We can use a similar approach to obtain an expression for the integral $I_{\E}$ from \S\ref{sec:mixed_elliptic_period} in terms of the same set of transcendental numbers. We find:
\beq\label{eq:I_E_res}
I_{\E} = -4 \, \pi \sqrt{3}\,  \Lambda(f,2)\, , 
\eeq
where
 \beq
\Lambda(f,2) = \int_{0}^{\infty}  f(it)t^2  \, \frac{dt}{t} = 0.85718907492991773071685111\ldots\,.
\eeq
was the completed $L$-value of $f$. 
Comparing~\eqref{eq:I_res} and~\eqref{eq:I_E_res} with ~\eqref{eq:relation} 
we find that 
\beq
I_{\mathrm{Pol}} = \frac{5}{\sqrt{3}}\,\textrm{Cl}_2\left(\frac{\pi}{3}\right)\,.
\eeq
While the results obtained here are based on high-precision numerical evaluations and the PSLQ algorithm, one can doubtless deduce an exact proof by viewing this integral as a double iterated integral of modular forms between cusps (remark \ref{remMMVinterp}), and applying the  Rankin-Selberg method to iterated integrals along the lines of \cite{Be} and  \cite{Brown:mmv2} \S9.

 %\section{Concluding remarks}

%\franciscomment{Possibly bin this or cannibalise for introduction?}   

 %The obstruction for a motivic period of an object $M$ to be polylogarithmic is if it admits a non-trivial infinitesimal action by a derivation $\sigma$ in the Lie algebra of the motivic Galois group, which is not of Artin-Tate type. Such derivations correspond to the existence of non-trivial extensions of pure  motives  inside $M$ which are not of Tate type. In the present situation, we could exhibit an extension of $H^1(\mathcal{E})$ by $\Q_{\chi}(-2)$ explicitly  
 %using a somewhat old-fashioned approach to computing the motive $M$. This clearly becomes % highly inpractical with more complicated examples, and one has to embrace the machinery  of motivic multiple modular values in order to compute the extension data more efficiently. 

%% file: modular_appendix.tex
% !TEX root = main.tex

\section{Differential forms on $\E\backslash\cC$ with logarithmic singularities}
\label{app:Eisenstein}

In this appendix we give the explicit expression for differential forms on $\E_k\backslash\cC$ with logarithmic singularities.
Since $\E_k\backslash\cC$ is an elliptic curve with 12 points removed, the first de Rham cohomology group of $\E_k\backslash\cC$ is generated by the classes of the holomorphic differential $-3\frac{dx}{y}$, a differential of the second kind, and 11 differentials with logarithmic singularities at the points of $\cC$. 

Under the modular parametrisation $\varphi:\Gamma(6)\backslash\mathfrak{H}\to \E\backslash\cC$ the holomorphic differential pulls back to the unique normalised cusp form $f(\tau)$ of weight two for $\Gamma(6)$, see~\eqref{eq:holomorphic_pullback}. The differentials with logarithmic singularities pull back to a linear combination of the cusp form and Eisenstein series of weight two. Every Eisenstein series of weight $n$ for $\Gamma(6)$ is a linear combination of the series in~\eqref{eq:EisensteinH_def}, and a linear independent set for $n=2$ is obtained for~\cite{Broedel:2018iwv}:
$$(r,s)\in\{(0,1),(0,2),(0,3),(1,0),(1,1),(1,2),(1,3),(1,4),(2,0),(2,2),(3,0)\}\,.$$

We now describe how to write the logarithmic differentials in terms of Eisenstein series and the cusp form of weight two.
As an example, let us consider the differential $\frac{dx}{y\,x}$, where we have chosen the positive branch of the square root so that $y=\sqrt{1+x^3}$. We have
$$
\varphi^*\frac{dx}{y\,x} = d\tau\,\frac{\partial_{\tau}x_6(\tau)}{y_6(\tau)\,x_6(\tau)} = -\frac{2\pi i}{3}\,d\tau\,\left[
q^3-5 q^9+6 q^{15}+8 q^{21}+\ldots\right]\,,
$$
with $q=e^{2\pi i\tau/6}$ and $x_6(\tau)$ and $y_6(\tau)$ are given in~\eqref{eq:x6y6_def}. By comparing the first few terms of this $q$-series to the $q$-expansion of a generic linear combination of $f(\tau)$ and a linear independent set of Eisenstein series of weight two, we find that
$$\varphi^*\frac{dx}{y\,x} = \frac{d\tau}{2\pi i}\,\left[  \frac{2}{3} H^{(2)}_{0,3}(\tau )-H^{(2)}_{1,0}(\tau )-2 H^{(2)}_{1,3}(\tau )+\frac{1}{3} H^{(2)}_{3,0}(\tau ) \right]\,.$$
All other cases can be obtained in a similar way, and we find:
\begin{align*}
\varphi^\ast\frac{dx}{x}&= \frac{d\tau}{2\pi i}\,\left[  3 H^{(2)}_{1,0}(\tau )+H^{(2)}_{3,0}(\tau ) \right]\,,\\
\varphi^\ast\frac{dx}{(x-2) y}&= \frac{d\tau}{2\pi i}\,\left[  -\frac{2}{3} H^{(2)}_{0,1}(\tau )-\frac{2}{9} H^{(2)}_{0,3}(\tau )+\frac{1}{3} H^{(2)}_{1,0}(\tau )+\frac{1}{3} H^{(2)}_{1,2}(\tau )+\frac{1}{3} H^{(2)}_{1,4}(\tau )\right.\\
&\phantom{=\frac{d\tau}{2\pi i}\,[}\left.-\frac{1}{9} H^{(2)}_{3,0}(\tau )-\frac{4\pi ^2}{9}  f(\tau ) \right]\,,\\
\varphi^\ast\frac{dx}{x-2}&= \frac{d\tau}{2\pi i}\,\left[  -H^{(2)}_{0,2}(\tau )+H^{(2)}_{1,0}(\tau )-H^{(2)}_{1,2}(\tau )-H^{(2)}_{1,4}(\tau )+H^{(2)}_{2,0}(\tau )\right.\\
&\phantom{=\frac{d\tau}{2\pi i}\,[}\left.+H^{(2)}_{3,0}(\tau ) \right]\,,\\
\varphi^\ast\frac{dx}{y (x+2 \rho)}&= \frac{d\tau}{2\pi i}\,\left[  \frac{4\pi ^2}{9}  \bar{\rho }\, f(\tau )-\frac{2}{9} H^{(2)}_{0,3}(\tau )+\frac{2}{3} H^{(2)}_{1,1}(\tau )+\frac{1}{3} H^{(2)}_{1,4}(\tau )-\frac{1}{9} H^{(2)}_{3,0}(\tau ) \right]\,,\\
\varphi^\ast\frac{dx}{x+2 \rho}&= \frac{d\tau}{2\pi i}\,\left[  2 H^{(2)}_{1,0}(\tau )+H^{(2)}_{1,4}(\tau )+H^{(2)}_{2,0}(\tau )-H^{(2)}_{2,2}(\tau )+H^{(2)}_{3,0}(\tau ) \right]\,,\\
\varphi^\ast\frac{dx}{y \left(x+2 \bar{\rho }\right)}&= \frac{d\tau}{2\pi i}\,\left[  \frac{4 \pi ^2}{9}\rho\,f(\tau )+\frac{2}{3} H^{(2)}_{0,1}(\tau )-\frac{2}{9} H^{(2)}_{0,3}(\tau )-\frac{2}{3} H^{(2)}_{1,0}(\tau )-\frac{2}{3} H^{(2)}_{1,1}(\tau )\right.\\
&\phantom{=\frac{d\tau}{2\pi i}\,[}\left.-\frac{1}{3} H^{(2)}_{1,2}(\tau )-\frac{2}{3} H^{(2)}_{1,3}(\tau )-\frac{2}{3} H^{(2)}_{1,4}(\tau )-\frac{1}{9} H^{(2)}_{3,0}(\tau ) \right]\,,\\
\varphi^\ast\frac{dx}{x+2 \bar{\rho }}&= \frac{d\tau}{2\pi i}\,\left[  H^{(2)}_{0,2}(\tau )+2 H^{(2)}_{1,0}(\tau )+H^{(2)}_{1,2}(\tau )+2 H^{(2)}_{2,0}(\tau )+H^{(2)}_{2,2}(\tau )\right.\\
&\phantom{=\frac{d\tau}{2\pi i}\,[}\left.+H^{(2)}_{3,0}(\tau ) \right]\,,\\
\varphi^\ast\frac{dx}{x+1}&= \frac{d\tau}{2\pi i}\,\left[  -H^{(2)}_{0,2}(\tau )+4 H^{(2)}_{1,0}(\tau )+2 H^{(2)}_{1,2}(\tau )+2 H^{(2)}_{1,4}(\tau )+H^{(2)}_{2,0}(\tau ) \right]\,,\\
\varphi^\ast\frac{dx}{x-\rho }&= \frac{d\tau}{2\pi i}\,\left[  2 H^{(2)}_{1,0}(\tau )-2 H^{(2)}_{1,4}(\tau )+H^{(2)}_{2,0}(\tau )-H^{(2)}_{2,2}(\tau ) \right]\,,\\
\varphi^\ast\frac{dx}{x-\bar{\rho }}&= \frac{d\tau}{2\pi i}\,\left[  H^{(2)}_{0,2}(\tau )+2 H^{(2)}_{1,0}(\tau )-2 H^{(2)}_{1,2}(\tau )+2 H^{(2)}_{2,0}(\tau )+H^{(2)}_{2,2}(\tau ) \right]\,.
\end{align*}